\documentclass[runningheads,envcountsame]{llncs}

\usepackage{ifluatex}
\ifluatex
  \usepackage{fontspec}
  \defaultfontfeatures{Ligatures=TeX}
\else
  \usepackage[utf8]{inputenc}
  \usepackage[T1]{fontenc}
\fi

\usepackage[english]{babel}
\usepackage[font=small,labelfont=bf,labelsep=period]{caption}
\usepackage[font=small,labelfont=normalfont]{subcaption}
\usepackage{wrapfig}
\usepackage{amsmath,amssymb,amsfonts}
\usepackage{graphicx}
\usepackage{enumerate}
\usepackage{xspace}
\usepackage{xparse}
\usepackage{multicol}
\usepackage{etoolbox}
\usepackage{xstring}
\DeclareMathOperator{\opt}{OPT}

\DeclareMathOperator{\alg}{ALG}

\newcommand{\hyper}{\ensuremath{\mathcal{H}}\xspace}
\newcommand{\hypersp}{\ensuremath{\mathcal{H}_f}\xspace}

\newcommand{\hedges}{\groups\xspace}
\newcommand{\hedgesp}{\groups_f\xspace}
\newcommand{\hedgespp}{\groups_p\xspace}
\newcommand{\hedgese}{E\xspace}

\newcommand{\hedgesepp}{E_p\xspace}

\newcommand{\groups}{\ensuremath{\Gamma}\xspace}
\newcommand{\group}{\ensuremath{g}\xspace}

\newcommand{\dmeet}{\ensuremath{\delta_{\operatorname{group}}}\xspace}
\newcommand{\dsep}{\ensuremath{\delta_{\operatorname{sep}}}\xspace}
\newcommand{\forbiddenh}{\ensuremath{\mathcal{F}}\xspace}
\newcommand{\ffree}{\forbiddenh-free\xspace}

\usepackage[algoruled,longend,vlined,procnumbered,nokwfunc]{algorithm2e}
\DontPrintSemicolon
\SetKw{KwBreak}{break}
\SetKw{KwTrue}{true}
\SetKw{KwFalse}{false}
\SetArgSty{textrm}
\usepackage[bookmarksopen,bookmarksdepth=3]{hyperref}

\newtheorem{RepeatTheoremEnvironment}{Theorem}
\newenvironment{repeattheorem}[1]%
    {\def\theRepeatTheoremEnvironment{\ref{#1}}%
     \begin{RepeatTheoremEnvironment}}%
    {\end{RepeatTheoremEnvironment}}
\newtheorem{RepeatLemmaEnvironment}{Lemma}
\newenvironment{repeatlemma}[1]%
    {\def\theRepeatLemmaEnvironment{\ref{#1}}%
     \begin{RepeatLemmaEnvironment}}%
    {\end{RepeatLemmaEnvironment}}
\newtheorem{RepeatObservationEnvironment}{Observation}
\newenvironment{repeatobservation}[1]%
    {\def\theRepeatObservationEnvironment{\ref{#1}}%
     \begin{RepeatObservationEnvironment}}%
    {\end{RepeatObservationEnvironment}}

\usepackage{tikz}
\usetikzlibrary{decorations.pathmorphing,decorations.pathreplacing,decorations.shapes,fit,positioning,shapes.geometric,calc,intersections,plotmarks}

\usepackage{arxiv-appendix}
\arxivid{1607.01196v2}    %
\arxivcite{ChaplickFLRVW16}    %
\arxivauxname{paper-arxiv}
\lncsfalse

\title{Block Crossings in Storyline Visualizations%
\iflncs%
  \thanks{\arxivrefthanks}
\else%
  \thanks{Appears in the Proceedings of the 24th International Symposium on
  Graph Drawing and Network Visualization (GD 2016).}
\fi%
}

\author{Thomas~C.~van~Dijk\inst{1} \and Martin Fink\inst{2} \and
  Norbert~Fischer\inst{1} \and
  Fabian~Lipp\inst{1} \and Peter~Markfelder\inst{1} \and
  Alexander~Ravsky\inst{3} \and Subhash Suri\inst{2} \and Alexander~Wolff\inst{1}}

\authorrunning{T.~C.~van~Dijk et al.}
\titlerunning{Block Crossings in Storyline Visualizations}

\institute{%
  Lehrstuhl f\"ur Informatik I, Universit\"at W\"urzburg, Germany\\
  \url{http://www1.informatik.uni-wuerzburg.de/en/staff} \and
  University of California, Santa Barbara, USA \and
  Pidstryhach Institute for Applied Problems of Mechanics and
  Mathematics,\\ National Academy of Science of Ukraine, Lviv, Ukraine}

\definecolor{cb-Set1-1}{RGB}{228,26,28}
\definecolor{cb-Set1-2}{RGB}{55,126,184}
\definecolor{cb-Set1-3}{RGB}{77,175,74}
\definecolor{cb-Set1-4}{RGB}{152,78,163}
\definecolor{cb-Set1-5}{RGB}{255,127,0}
\definecolor{cb-Set1-6}{RGB}{255,255,51}
\definecolor{cb-Set1-7}{RGB}{166,86,40}
\definecolor{cb-Set1-8}{RGB}{247,129,191}
\definecolor{cb-Set1-9}{RGB}{153,153,153}
\definecolor{cb-Dark2-1}{RGB}{27,158,119}
\definecolor{cb-Dark2-2}{RGB}{217,95,2}
\definecolor{cb-Dark2-3}{RGB}{117,112,179}
\definecolor{cb-Dark2-4}{RGB}{231,41,138}
\definecolor{cb-Dark2-5}{RGB}{102,166,30}
\definecolor{cb-Dark2-6}{RGB}{230,171,2}
\definecolor{cb-Dark2-7}{RGB}{166,118,29}
\definecolor{cb-Dark2-8}{RGB}{102,102,102}

\tikzset{
  meeting/.append style={ultra thick},
  char/.append style={very thick},
  char1/.append style={char,cb-Dark2-1},
  char2/.append style={char,cb-Dark2-2},
  char3/.append style={char,cb-Dark2-3},
  char4/.append style={char,cb-Dark2-4},
  char5/.append style={char,cb-Dark2-5},
  char6/.append style={char,cb-Dark2-6},
  char7/.append style={char,cb-Dark2-7},
  char8/.append style={char,cb-Dark2-8},
}

\newcounter{x}
\makeatletter
\providecommand*\@nameundef[1]{%
  \expandafter\let\csname #1\endcsname\@undefined}
\pgfkeys{
  /storyline/.is family, /storyline,
  drawingstyle/.estore in = \storyline@drawingstyle,
  leftwidth/.estore in = \storyline@leftwidth,
  rightwidth/.estore in = \storyline@rightwidth,
  crossinglength/.estore in = \storyline@crossinglength,
  distx/.estore in = \storyline@distx,
  disty/.estore in = \storyline@disty,
  label/.estore in = \storyline@label,
  labelshifty/.estore in = \storyline@labelshifty,
  name path/.initial,
  l1/.style={drawingstyle=char1,label=$1$},
  l2/.style={drawingstyle=char2,label=$2$},
  l3/.style={drawingstyle=char3,label=$3$},
  l4/.style={drawingstyle=char4,label=$4$},
  l5/.style={drawingstyle=char5,label=$5$},
  l6/.style={drawingstyle=char6,label=$6$},
  l7/.style={drawingstyle=char7,label=$7$},
  l8/.style={drawingstyle=char8,label=$8$},
  invisible/.code={\def\storyline@invisible{true}},
  norightlabel/.code={\def\storyline@norightlabel{true}},
  default/.style = {
    drawingstyle={},
    leftwidth={0},
    rightwidth={0},
    crossinglength={0.5},
    distx={6},
    disty={0.5},
    label={},
    labelshifty={}
  }
}
\newcommand*{\storylineset}{\pgfqkeys{/storyline}}
\storylineset{default}
\newcommand{\drawstoryline}[2][]{
  \begin{scope}[/storyline, #1]
  \setcounter{x}{0}
  \@nameundef{lasty}
  \@nameundef{firsty}
  \@nameundef{firstx}
  \xdef\pathtodraw{}
  \foreach \y in {#2} {
    \IfStrEq{\y}{x}{
      \stepcounter{x}
    }{
      \@ifundefined{firstx}{
        \xdef\firstx{\thex}
      }{}
      \@ifundefined{lasty}{
        \xdef\lasty{\y}
        \xdef\firsty{\y}
        \IfStrEq{\storyline@leftwidth}{x}{}{
          \xappto\pathtodraw{({\storyline@distx*(\thex-\storyline@leftwidth)},
              \storyline@disty*\lasty)
            -- (\storyline@distx*\thex,\storyline@disty*\lasty)}
        }
      }{
        \IfStrEq{\pathtodraw}{}{
          \xappto\pathtodraw{(\storyline@distx*\thex,\storyline@disty*\lasty)}
        }{}
        \xappto\pathtodraw{-- (\storyline@distx*\thex,\storyline@disty*\lasty)
          .. controls ({\storyline@distx*(\thex+\storyline@crossinglength/2)},
              \storyline@disty*\lasty)
          and ({\storyline@distx*(\thex+\storyline@crossinglength/2)},
              \storyline@disty*\y)
          .. ({\storyline@distx*(\thex+\storyline@crossinglength)},\storyline@disty*\y)}
        \xdef\lasty{\y}
        \stepcounter{x}
      }
    }
  }
  \IfStrEq{\storyline@rightwidth}{x}{}{
    \xappto\pathtodraw{-- ++(\storyline@rightwidth*\storyline@distx,0)}
  }
  \IfStrEq{\pgfkeysvalueof{/storyline/name path}}{}{
    \@ifundefined{storyline@invisible}{
      \draw[\storyline@drawingstyle] \pathtodraw;
    }{
      \path[] \pathtodraw;
    }
  }{
    \@ifundefined{storyline@invisible}{
      \draw[\storyline@drawingstyle,
          name path global={\pgfkeysvalueof{/storyline/name path}}]
        \pathtodraw;
    }{
      \path[name path global={\pgfkeysvalueof{/storyline/name path}}] \pathtodraw;
    }
  }

  \@ifundefined{storyline@invisible}{
    \xdef\storyline@shiftstr{}
    \IfStrEq{\storyline@labelshifty}{}{}{
        \xdef\storyline@shiftstr{+ (0,\storyline@labelshifty)}
    }
    \IfStrEq{\storyline@leftwidth}{x}{}{
      \xdef\storyline@pos{($ ({\storyline@distx*(\firstx-\storyline@leftwidth)},
          \storyline@disty*\firsty) - (.3333em,0) \storyline@shiftstr $)}
      \node[\storyline@drawingstyle,anchor=east,inner sep=0]
        at \storyline@pos
        {\storyline@label};
    }
    \IfStrEq{\storyline@rightwidth}{x}{}{
      \@ifundefined{storyline@norightlabel}{
        \xdef\storyline@pos{($ ({\storyline@distx*
            (\thex-1+\storyline@crossinglength+\storyline@rightwidth)},
            \storyline@disty*\lasty) + (.3333em, 0) \storyline@shiftstr $)}
        \node[\storyline@drawingstyle,anchor=west,inner sep=0]
          at \storyline@pos
          {\storyline@label};
      }{}
    }
  }{}
  \end{scope}
}
\newcommand{\drawstorylinemeeting}[4][]{
  \begin{scope}[/storyline, #1]
  \draw[meeting,#1] (\storyline@distx*#2,{\storyline@disty*(#3-0.3)})
    -- (\storyline@distx*#2,{\storyline@disty*(#4+0.3)});
  \end{scope}
}
\newcommand{\drawstorylinepermutation}[5][]{
  \begin{scope}[/storyline, #1]
    \node[anchor=south,yshift=0.2cm] at (#2*\storyline@distx,#4*\storyline@disty) {#5};
    \draw[dashed] (#2*\storyline@distx,{(#3-0.5)*\storyline@disty})
      -- (#2*\storyline@distx,{(#4+0.5)*\storyline@disty});
  \end{scope}
}
\makeatother

\let\doendproof\endproof
\renewcommand\endproof{~\hfill$\qed$\doendproof}

\newcommand{\CS}{\ensuremath{\mathcal{S}}\xspace}
\newcommand{\SBT}{SBT\xspace}

\newtheorem{observation}[theorem]{Observation}

\newcommand{\algo}[1]{\textsc{#1}}

\begin{document}

\maketitle

\begin{abstract}
  Storyline visualizations help visualize encounters of the
  characters in a story over time.  Each character is represented by
  an $x$-monotone curve that goes from left to right.  A meeting is
  represented by having the characters that participate in the meeting
  run close together for some time.  In order to keep the visual
  complexity low, rather than just minimizing pairwise crossings of
  curves, we propose to count \emph{block crossings}, that is, pairs
  of intersecting bundles of lines.

  Our main results are as follows.  We show that minimizing the number
  of block crossings is NP-hard, and we develop, for meetings of
  bounded size, a constant-factor approximation.  We also present two
  fixed-parameter algorithms and, for meetings of size~2, a greedy
  heuristic that we evaluate experimentally.
\end{abstract}

\section{Introduction}
\label{sec:introduction}

A storyline visualization is a convenient abstraction for visualizing
the complex narrative of interactions among people, objects, or concepts. 
The motivation comes from the setting of a movie, novel, or play where the 
narrative develops as a sequence of interconnected scenes, each involving a 
subset of characters. See Fig.~\ref{fig:jurassic-park-xkcd} for an example.

\begin{figure}[b]
  \centering
  \includegraphics[width=\textwidth]{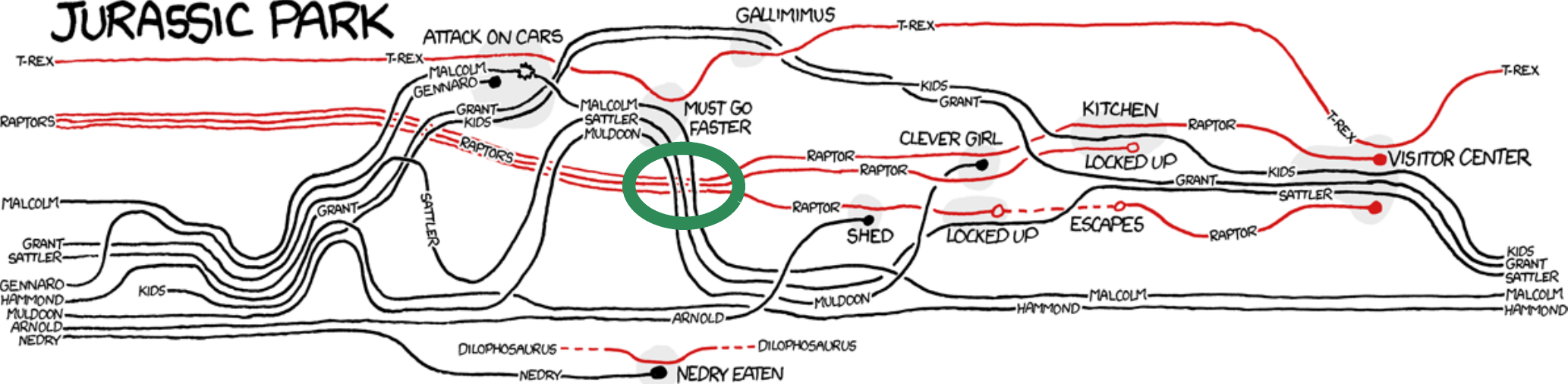}
  \caption{Storyline visualization for \emph{Jurassic Park}
  by \texttt{xkcd}~\cite{xkcd-storylines} with a block crossing
  (highlighted by a bold green ellipse).}
  \label{fig:jurassic-park-xkcd}
\end{figure}
The storyline abstraction of characters and events occurring over time
can be used as a metaphor for visualizing other situations, from
physical events involving groups of people meeting in corporate
organizations, political leaders managing global affairs, and groups of
scholars collaborating on research to abstract co-occurrences of
``topics'' such as a global event being covered on the front pages of
multiple leading news outlets, or different organizations 
turning their attention to a common cause.

A storyline visualization maps a set of characters
of a story to a set of curves in the plane and a sequence of meetings
between the characters to regions in the plane where the corresponding
curves come close to each other.
The current form of storyline visualizations seems to have been
invented by Munroe~\cite{xkcd-storylines} (compare
Fig.~\ref{fig:jurassic-park-xkcd}), who used it to
visualize, in a compact way, which subsets of characters meet over the
course of a movie.
Each character is shown as an x-monotone curve. 
Meetings occur at certain times from left to right. 
A meeting corresponds to a point in time where the characters that
meet are next to each other with only small gaps between them. Munroe
highlights meetings by underlaying them with a gray shaded region,
while we use a vertical line for that purpose.    Hence, a storyline
visualization can be seen as a drawing of a hypergraph whose vertices
are represented by the curves and whose edges come in at specific
points in time.

A natural objective for the quality of a storyline visualization is 
to minimize unnecessary ``crossings'' among the character lines.
The number of crossings alone, however, is a poor measure: two blocks 
of ``locally parallel'' lines crossing each other are far less 
distracting than an equal number of crossings randomly scattered 
throughout the drawing.  Therefore, instead of pairwise crossings, 
we focus on minimizing the number of \emph{block crossings}, where
each block crossing involves two arbitrarily large sets of parallel
lines forming a crossbar, with no other line in the crossing area;
see Fig.~\ref{fig:jurassic-park-xkcd} for an example.

\paragraph{Previous Work.}
Kim et al.~\cite{kch-tgdt-AVI10} used
storylines to visualize genealogical data; meetings correspond to
marriages and special techniques are used to indicate child--parent
relationships.  Tanahashi and Ma~\cite{tm-dcosv-TVCG12} computed
storyline visualizations automatically and showed how to adjust the
geometry of individual lines to improve the aesthetics of their
visualizations.  Muelder et al.~\cite{mcsm-esval-BD13} visualized
clustered, dynamic graphs as storylines, summarizing the behavior of
the local network surrounding user-selected foci.

Only recently a more theoretical and principled study was initiated by
Kostitsyna et al.~\cite{DBLP:conf/gd/KostitsynaNP0S15}, who considered
the problem of minimizing pairwise (not \emph{block}) crossings in storylines. 
They proved that the problem is NP-hard in general, and showed that it is
fixed-parameter tractable with respect to the (total) number of characters.
For the special case of 2-character meetings without repetitions, they developed 
a lower bound on the number of crossings, as well as as an upper bound of 
$O(k \log k)$ when the meeting graph---whose edges describe the pairwise
meetings of characters---is a tree.

Our work builds on the problem formulation of Kostitsyna et
al.~\cite{DBLP:conf/gd/KostitsynaNP0S15} but we considerably extend
their results by designing (approximation) algorithms for
general meetings---for a different optimization goal: we minimize the
number of \emph{block crossing} rather than the number of pairwise
line crossings.  Block crossings were introduced by Fink et
al.~\cite{DBLP:journals/jgaa/FinkPW15} for visualizing metro maps.

\paragraph{Problem Definition.}
A storyline $\CS$ is a pair $(C,M)$ where $C = \{1, \dots, k\}$ is a set of
\emph{characters} and $M = [m_1, m_2, \dots, m_n]$ with $m_i
\subseteq C$ and $|m_i| \ge 2$ for $i = 1, 2, \dots, n$ is a
sequence of \emph{meetings} of at least two characters. We call any
set $\group \subseteq C$ of characters that has at least one meeting,
a \emph{group}. We define the \emph{group hypergraph} $\hyper = (C,\hedges)$
whose vertices are the characters and whose hyperedges are the groups
that are involved in at least one meeting. The group hypergraph does
not include the temporal aspect of the storyline---it models only the
graph-theoretical structure of groups participating in the storyline
meetings; it can be built by lexicographically sorting the
meetings in $M$ in $O(nk\log n)$ time.

Note that we do not encode the exact times of the meetings: In a given
visualization, at any time $t$, there is a unique vertical order $\pi$
of the characters.  Without changing $\pi$ by crossings, we can
increase or decrease vertical
gaps between lines. If a group $\group$ forms a contiguous
interval in $\pi^t$, then we can bring $\group$'s lines within a short
distance $\dmeet$ without any crossing, and also make sure that all
other lines are at a larger distance of at least $\dsep$. Since any
group must be supported at a time just before its meeting
starts, computing an output drawing consists mainly of changing the
permutation of characters over time so that during a meeting its group
is supported by the current permutation. We therefore focus on
changing the permutation by crossings over time, and only have to be
concerned about the order of meetings; the final drawing can be
obtained by a simple post-processing from this discrete set of
permutations.

\begin{wrapfigure}[9]{r}{.35\textwidth}
  \centering
  \begin{tikzpicture}[scale=0.7]
    \storylineset{disty=0.35}
      \storylineset{
        norightlabel,
        leftwidth=0.05,
        rightwidth=0.05,
        disty=0.35,
      }
    \drawstoryline[drawingstyle=char1,label=$k$]{3,3}
    \drawstoryline[drawingstyle=char1]{4,4}

    \drawstoryline[drawingstyle=char2,label=$c$]{5,9}
    \drawstoryline[drawingstyle=char2]{6,10}
    \drawstoryline[drawingstyle=char2,label=$b{+}1$,labelshifty=-2pt]{7,11}

    \drawstoryline[drawingstyle=char3,label=$b$,labelshifty=2pt]{8,5}
    \drawstoryline[drawingstyle=char3]{9,6}
    \drawstoryline[drawingstyle=char3]{10,7}
    \drawstoryline[drawingstyle=char3,label=$a$]{11,8}

    \drawstoryline[drawingstyle=char1]{12,12}
    \drawstoryline[drawingstyle=char1,label=$1$]{13,13}

  \end{tikzpicture}
    \caption{Block crossing $(a,b,c)$}
    \label{fig:block-move}
\end{wrapfigure}
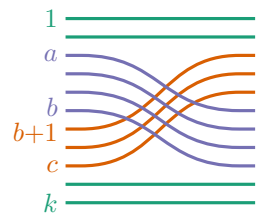
If $\{\pi_1,\pi_2,\dots,\pi_k\}=\{1,2,\dots,k\}$, then $\langle \pi_1,
\pi_2, \dots, \pi_k \rangle$ is a permutation of length~$k$ of~$C$.
For $a \le b < c$, a \emph{block crossing} $(a, b, c)$ on the permutation
$\pi = \langle 1, \dots, k \rangle$ is the exchange of
two consecutive blocks $\langle a,\dots,b \rangle$ and
$\langle b+1,\dots,c \rangle$; %
see Fig.~\ref{fig:block-move}. %
A meeting~$m$ \emph{fits} a permutation~$\pi$ (or a
permutation~$\pi$ \emph{supports} a meeting~$m$) if the characters participating
in~$m$ form an interval in~$\pi$.  In other words, there is a
permutation of~$m$ that is part of~$\pi$.  If we apply a sequence~$B$
of block crossings to a permutation~$\pi$ in the given order, we denote
the resulting permutation by~$B(\pi)$.

\begin{problem}[Storyline Block Crossing Minimization (SBCM)]
  Given a storyline instance $(C, M)$ find a solution consisting of a start
  permutation~$\pi^0$ of~$C$ and a 
  sequence~$B$ of (possibly empty) sequences of block crossings
  $B_1,B_2,\dots,B_n$ such that the total number of block crossings is
  minimized and~$\pi^1=B_1(\pi^0)$ supports~$m_1$,
  $\pi^2=B_2(\pi^1)$ supports~$m_2$, etc.
\end{problem}
We also consider $d$-SBCM, a special case of SBCM where meetings
involve groups of size at most $d$, for an arbitrary constant~$d$.
E.g., 2-SBCM allows only 2-character meetings, a setting that
was also studied by Kostitsyna et
al.~\cite{DBLP:conf/gd/KostitsynaNP0S15}.

\paragraph{Our Results.}
We observe that a storyline has a crossing-free visualization if and
only if its group hypergraph is an interval hypergraph.  A hypergraph can be
tested for the interval property in $O(n^2)$ time, where $n$ is the
number of hyperedges.
We show that 2-SBCM is NP-hard (see Sect.~\ref{sec:np-hard}) and
that SBCM is fixed-parameter tractable with respect to~$k$
(Sect.~\ref{sec:exact}).
The latter can be modified to handle pairwise crossings, where its runtime improves on Kostitsyna et al.~\cite{DBLP:conf/gd/KostitsynaNP0S15}.

We present a greedy
algorithm for 2-SBCM that runs in $O(k^3n)$ time for $k$ characters.
We do some preliminary experiments where
we compare greedy solutions to optimal solutions; see
Sect.~\ref{sec:algorithm}.
One of our main results is a constant-factor
approximation algorithm for $d$-SBCM for the case that $d$ is bounded
and that meetings cannot be repeated; see Sect.~\ref{sec:approximation}. Our
algorithm is based on a solution for the following NP-complete
hypergraph problem, which may be of independent interest.
Given a
hypergraph $\hyper$, we want to delete the minimum number of
hyperedges so that the remainder is an interval hypergraph. We
develop a $(d+1)$-approximation algorithm, where $d$ is the maximum
size of a hyperedge in $\hyper$; see Sect.~\ref{sec:hyperedge-removal-approx}.
Finally, we list some open problems in Appendix~\appref{app:openProblems}.

\section{Preliminaries}
\label{sec:preliminaries}

First, we consider the special case where every meeting
consists of two characters.  For these restricted instances, every
meeting can be realized from any permutation by a single
block crossing.
This raises the question
whether there is also an \emph{optimal} solution that fulfills this
condition.  The answer is negative---if we may
prescribe the start permutation;
see Appendix~\appref{app:preliminaries} for details.

\newcommand{\contentObsBadupper}{%
  Given an instance of 2-SBCM, there is a solution with at most
  one block crossing before each of the meetings.
  In particular, there is a solution with at most $n$ block crossings in total.
}
\begin{observation}
  \label{obs:badupper}
  \contentObsBadupper
\end{observation}

\paragraph{Detecting Crossing-Free Storylines.}
If a storyline admits a crossing-free visualization, then
the vertical permutation of the character lines remains 
the same over time, and all meetings involve groups that form contiguous 
subsets in that permutation. (The visualization
can be obtained by placing characters along a vertical line in the correct 
permutation and for each meeting bringing its lines together for the 
duration of the meeting and then separating them apart again.)
In other words, a single permutation supports each group of
$\hyper = (C,\hedges)$. This holds if and only if $\hyper$ is an
\emph{interval hypergraph}.  This is the case
if there exists a permutation $\pi =\left\langle v_1,
\ldots, v_k \right\rangle$ 
of~$C$ such that each hyperedge $e \in \hedges$ corresponds to 
a contiguous block of characters in this permutation.
As an anonymous reviewer pointed out, this is equivalent to the
hypergraph having path support~\cite{DBLP:journals/jgaa/BuchinKMSV11}.
An interval hypergraph can be visualized by placing all of its vertices on 
a line, and drawing each of its hyperedges as an interval that includes
all vertices of $e$ and no vertex of $V \setminus e$.
Checking whether a $k$-vertex hypergraph is an interval hypergraph
takes $O(k^2)$ time %
\cite{trotter1976characterization}.
Recall that we can build $\hyper$ in $O(nk \log n)$ time.

\begin{theorem}
  \label{obs:crossing-free}
	Given the group hypergraph $\hyper$ of an instance of SBCM with $k$
	characters, we can check in $O(k^2)$ time whether a crossing-free
	solution exists.
\end{theorem}

For 2-SBCM we only need to check (in $O(k)$ time) whether $\hyper$ 
is a collection of vertex-disjoint paths; this
is dominated by the time ($O(n)$) for building $\hyper$.

\section{NP-Completeness of SBCM}
\label{sec:np-hard}

In this section we prove that SBCM is NP-complete.  This is known for
BCM.  But SBCM is not simply a generalization of BCM because in SBCM
we can choose an arbitrary start permutation.  Therefore, the idea of
our hardness proof is to force a certain start permutation by adding
some characters and meetings.
We reduce from \textsc{Sorting by Transpositions} (\SBT), which has also 
been used to show the hardness of BCM~\cite{DBLP:journals/jgaa/FinkPW15}.  
In \SBT, the problem is to decide whether there is a sequence of
transpositions (which are equivalent to block crossings) of length 
at most $k$ that transforms a given permutation $\pi$ to the identity.
\SBT was recently shown NP-hard by Bulteau et al.~\cite{bfr-sbtid-DM12}.

We show hardness for 2-SBCM, which also implies that SBCM is NP-hard.
It is easy to see that SBCM is in NP:
Obviously, the maximum number of block crossings needed for any number of
characters and meetings is bounded by a polynomial in $k$ and $n$.
Therefore also the size of the solutions is bounded by a polynomial.
To test the feasibility of a solution efficiently, we simply test whether the
permutations between the block crossings support the meetings in the right order
from left to right.  We will use the following obvious fact.

\begin{observation}
  \label{obs:permutationsolve}
  If permutation $\pi$ needs $c$ block crossings to be sorted,
  any permutation containing~$\pi$ as subsequence needs at least
  $c$ block crossings to be sorted.
\end{observation}

\begin{theorem}
2-SBCM is NP-complete.
\end{theorem}
\noindent\emph{Proof.}\nobreakspace
It remains to show the NP-hardness.  We reduce from \SBT.
Given an instance of \SBT, that is, a permutation $\pi$ of $\{1,
\dots, k\}$, we show how to use a hypothetical, efficient algorithm for
2-SBCM to determine the minimum number of transpositions (i.e., block crossings) that
transforms~$\pi$ to the identity $\iota=\langle 1, 2, \dots, k \rangle$.
Note that $\pi$ can be sorted by
at most $k$ block crossings.
So $k$ is an upper bound for an optimal solution of instance $\pi$ of \SBT.

We extend the set of characters~$\{1,2,\dots,k\}$ to $C = \{1, \dots,
k, c_1, c_2, \dots, c_{2k}\}$.  Correspondingly, we
extend~$\pi=\langle \pi_1, \pi_2, \dots, \pi_k \rangle$ to
$\pi'=\langle c_1, \dots, c_{2k}, \pi_1, \dots, \pi_k \rangle$
and~$\iota$ to $\iota' = \langle c_1, c_2, \dots, c_{2k}, 1, 2, \dots, k
\rangle$.
Let~$M_{\pi'}$ and~$M_{\iota'}$ be the sequences of meetings of all
neighboring pairs in~$\pi'$ and~$\iota'$, respectively.
Let $M_1$ and $M_2$ be the concatenations of $k+1$ copies of $M_{\pi'}$ and
$M_{\iota'}$, respectively.
By repeating we get $M_1 = M_{\pi'}^{k+1}$ and $M_2 = M_{\iota'}^{k+1}$.
This yields the instance $\CS = (C, M)$ of 2-SBCM, where $M$ is the
concatenation of~$M_1$ and~$M_2$; see Fig.~\ref{fig:np-reduction}.

We show that the number of block crossings needed for the 2-SBCM instance~\CS
equals the number of block crossings to solve instance~$\pi$ of \SBT.

First, let $B$ be a shortest sequence of block crossings to sort~$\pi$.
Then, $(\pi', B)$ is a feasible solution for $\CS$.
The start permutation $\pi'$ supports all meetings in $M_1$ without any
block crossing.  Using~$B$, the lines are sorted to $\iota'$, and this
permutation supports all meetings in~$M_2$ without any further block
crossings; see Fig.~\ref{fig:np-reduction}.
Hence, the number of block crossings in any solution of~$\pi$ is an
upper bound for the minimum number of block crossings needed for~\CS.

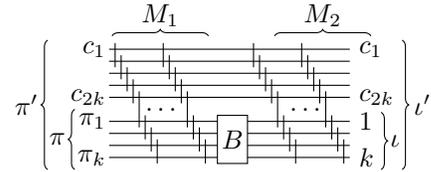
\begin{wrapfigure}[11]{r}{0.46\textwidth}
  \centering
  \begin{tikzpicture}[scale=0.8]
    \foreach \y in {3.6,3.4,3.2,3.0,2.8,2.4,2.2,...,1.8} {
      \draw (0.1,\y) -- (4.1,\y);
    }
    \node[anchor=east] at (0.2,3.6) {$c_1$};
    \node[anchor=east] at (0.2,2.8) {$c_{2k}$};
    \node[anchor=east] at (0.2,2.4) {$\pi_1$};
    \node[anchor=east] at (0.2,1.8) {$\pi_k$};
    \node[anchor=west] at (4.1,3.6) {$c_1$};
    \node[anchor=west] at (4.1,2.8) {$c_{2k}$};
    \node[anchor=west] at (4.1,2.4) {$1$};
    \node[anchor=west] at (4.1,1.8) {$k$};

    \foreach \x in {0, 0.8, 2.3, 3.1} {
      \foreach \i in {2,3,...,5} {
        \draw (\x + \i*0.1, 4 - \i*0.2 + 0.1) -- (\x + \i*0.1, 4 - \i*0.2 - 0.3);
      }
      \draw (\x + 0.6, 2.9) -- (\x + 0.6, 2.3);
      \foreach \i in {7,...,9} {
        \draw (\x + \i*0.1, 3.8 - \i*0.2 + 0.1) -- (\x + \i*0.1, 3.8 - \i*0.2 - 0.3);
      }
    }
    \draw[fill=white] (1.9,2.5) rectangle (2.4,1.7);
    \node at (2.15,2.1) {$B$};

    \draw[decoration={brace,raise=5pt},decorate] (0.15,3.6) -- node[above=5pt] {$M_1$} (1.75,3.6);
    \draw[decoration={brace,raise=5pt},decorate] (2.85,3.6) -- node[above=5pt] {$M_2$} (4.45,3.6);
    \draw[decoration={brace,raise=5pt},decorate] (-0.25,1.6) -- node[left=5pt] {$\pi$} (-0.25,2.6);
    \draw[decoration={brace,raise=5pt},decorate] (-0.7,1.6) -- node[left=5pt] {$\pi'$} (-0.7,3.75);
    \draw[decoration={brace,raise=5pt},decorate] (4.4,2.55) -- node[right=5pt] {$\iota$} (4.4,1.6);
    \draw[decoration={brace,raise=5pt},decorate] (4.75,3.75) -- node[right=6pt] {$\iota'$} (4.75,1.6);
    \node at (1.0,2.6) {$\dots$};
    \node at (3.4,2.6) {$\dots$};
  \end{tikzpicture}
  \caption{Solution for the 2-SBCM instance $\CS$ corresponding to a solution
    $B$ of instance $\pi$ of \SBT.
    The box $B$ represents the block crossings.}
  \label{fig:np-reduction}
\end{wrapfigure}
For the other direction, let $(\pi^*, B^*)$ be an optimal solution for $\CS$.
Any solution of 2-SBCM gives rise to a symmetric solution that is obtained by
reversing the order of the characters.
Without loss of generality, we assume that $\pi'$ (rather than the
reverse permutation $\pi^{\prime R}$) occurs somewhere in~$M_1$.

Next, we show that the start permutation~$\pi'$ occurs somewhere
in~$M_1$ and that $\iota'$ occurs somewhere in~$M_2$.
If there is a sequence~$M_{\pi'}$
of meetings between which there is no block crossing, the permutation
at this position can only be the start permutation $\pi'$ or its reverse.
For a contradiction, assume that $\pi'$ does not occur during~$M_1$ in the
layout induced by $(\pi^*, B^*)$.
Then there is no such sequence without any block crossing in it.
As this sequence is repeated $k+1$ times, the solution would need at least
$k+1$ block crossings.  This contradicts our upper bound, which is~$k$.
Analogously, we can show that the permutation~$\iota'$ or its reverse
occurs in~$M_2$.

We now want to show that the unreversed version of $\iota'$ occurs
in~$M_2$.
For a contradiction, assume the opposite.
We forget about the lines $1, \dots, k$ and only consider the sequence $\pi'' =
\langle c_1, \dots, c_{2k} \rangle$ in $\pi'$ which is reversed to $\iota''= \langle c_{2k}, \dots,
c_1 \rangle$ in $\iota^{\prime R}$.
Eriksson et al.~\cite{eriksson2001sorting} showed that we need
$\lceil (l+1)/2 \rceil$ block crossings to reverse a permutation of $l$
elements.
This implies that we need $k+1$ block crossings to transform $\pi''$ to $\iota''$.
As $\pi'$ and $\iota^{\prime R}$ contain these sequences as subsequences,
Observation~\ref{obs:permutationsolve} implies that the
transformation from $\pi'$ to $\iota^{\prime R}$ also needs at least $k+1$ block
crossings.
As the optimal solution uses at most $k$ block crossings, we know that we cannot
reach $\iota^{\prime R}$ and thus the sequence of permutations contains $\pi'$
and $\iota'$.

The sequence of block crossings that transforms~$\pi'$ to~$\iota'$ yields
a sequence~$B$ of block crossings of the same length that
transforms~$\pi$ to~$\iota$.  This
shows that the length of a solution for $\CS$ is an upper bound for
the length of an optimal solution of the corresponding \SBT
instance~$\pi$.  Thus, the two %
are equal.
\qed

\paragraph{Hardness Without Repetitions.}
With arbitrarily large meetings, SBCM is
hard even without repeating meetings. We can emulate a repeated
sequence of 2-char\-ac\-ter meetings by gradually increasing group sizes; see
Appendix~\appref{sec:hardness-no-rep}.

\section{Exact Algorithms}
\label{sec:exact}

\newcommand{\method}[1]{\textsc{#1}}

We present two exact algorithms.
Conceptually, both build up a sequence of block crossings
while keeping track of how many meetings have already been accomplished.
The first uses polynomial space; the second improves the runtime at
the cost of exponential space.

We start with a data structure that keeps track of permutations,
block crossings and meetings.
It is initialized with a given permutation and has two operations.
The \method{Check} operation returns whether a given meeting fits the current permutation.
The \method{BlockMove} operation performs a given block crossing on the permutation and then returns whether the most-recently \method{Check}ed meeting now fits.
See Appendix~\appref{app:exact} for a detailed description.

\newcommand{\contentLemDatastruct}{%
A sequence of arbitrarily interleaved \method{BlockMove} and
\method{Check} operations can be performed in $O(\beta+\mu)$ time,
where $\beta$ is the number of block crossings and $\mu$ is sum of cardinalities
of the meetings given to \method{Check}.
Space usage is $O(k)$.
}
\begin{lemma}
  \label{lem:datastruct}
  \contentLemDatastruct
\end{lemma}
A block crossing can be represented by indices
$(a,b,c)$ with $1\leq a\leq b<c\leq k$; hence, there are
$\frac{k^3-k}{6}$ distinct block crossings on a permutation of length
$k$.

Now we provide an output-sensitive algorithm for SBCM whose runtime
depends on the number of block crossings required by the optimum.

\begin{theorem}
\label{thm:polySpaceExact}
An instance of SBCM can be solved in $O(k!\cdot (\frac{k^3-k}{6})^\beta\cdot (\beta+\mu))$ time and $O(\beta k)$ working space if a solution with $\beta$ block crossings exists, where $\mu=\sum_{i\in M} |m_i|$.
\end{theorem}
\begin{proof}
Consider a branching algorithm that starts from a permutation of the
characters and keeps trying all possible block crossings.
This has branching factor $\frac{k^3-k}{6}$ and we can enumerate the
children of a node in constant time each by enumerating triples
$(a,b,c)$.
While applying block crossings, the algorithm keeps track of how many meetings fit this sequence of permutations using the data structure from Lemma~\ref{lem:datastruct}.
We use depth-first iterative-deepening search~\cite{k-dfid-AI85} from
all possible start permutations until we find a sequence of
permutations that fulfills all meetings.  Correctness follows from the
iterative deepening: we want an (unweighted) shortest sequence of block crossings.
The runtime and space bounds follow from the standard analysis of iterative-deepening search, observing that a node uses $O(k)$ space and it takes $O(\beta+\mu)$ time in total to evaluate a path from root to leaf.
\end{proof}
We have that $\mu$ is $O(kn)$ since there are $n$ meetings and each
consists of at most $k$ characters.
At the cost of exponential space, we can improve the runtime and get
rid of the dependence on~$\beta$, showing the problem to be fixed
parameter linear for $k$.
We note that the following algorithm can easily be adapted to handle
pairwise crossings rather than block crossings; in this case the
runtime improves upon the original result of Kostisyna et
al.~\cite{DBLP:conf/gd/KostitsynaNP0S15} by a factor of~$k!$.

\begin{theorem}
\label{thm:fpt}
An instance of SBCM can be solved in $O( k!\cdot k^3\cdot n )$ time and $O( k!\cdot k \cdot n )$ space.
\end{theorem}
\begin{proof}

Let $f(\pi, \ell)$ be the optimal number of block crossings in a solution to the given instance when restricted to the first $\ell$ meetings and to have $\pi$ as its final permutation.
Note that by definition the solution for the actual instance is given by $\min_{\pi^*} f(\pi^*,n)$, where the minimum ranges over all possible permutations.
As a base case, $f(\pi,0)=0$ for all $\pi$, since the empty set of meetings is supported by any permutation.
Let $\pi$ and $\pi'$ be permutations that are one block crossing apart and let $0\leq \ell \leq \ell'$.
If the meetings $\{m_{\ell+1}, \ldots, m_{\ell'}\}$ fit $\pi'$, then
$f(\pi',\ell')\leq f(\pi,\ell)+1$: if we can support the first $\ell$
meetings and end on $\pi$, then with one additional block crossing we can
support the first $\ell'$ meetings and end with $\pi'$.

We now model this as a graph.
Let $G$ be an unweighted directed graph on nodes $(\pi,\ell)$ and call a node \emph{start node} if $\ell=0$. %
There is an arc from $(\pi,\ell)$ to $(\pi',\ell')$ if and only if
$\pi$ and $\pi'$ are one block crossing apart, $\ell\leq\ell'$, and the meetings $\{m_{\ell+1}, \ldots, m_{\ell'}\}$ fit $\pi'$.
Note that we allow $\ell=\ell'$ since we may need to allow block
crossings that do not immediately achieve an additional meeting (cf. Proposition~\appref{prop:counterexample}), so $G$ is not acyclic.
In the constructed graph, $f(\pi,\ell)$ equals the graph distance from the node $(\pi,\ell)$ to the closest start node.
Call a path to a start node that realizes this distance \emph{optimal}.

In $G$, consider any path $[(\pi_1,\ell_1),(\pi_2,\ell_2),(\pi_3,\ell_3)]$ with $\ell_3>\ell_2$.
If meeting $\ell_2+1$ fits $\pi_2$, then $[(\pi_1,\ell_1),(\pi_2,\ell_2+1),(\pi_3,\ell_3)]$ is also a path.
Repeating this transformation shows that for all $\pi$, the node $(\pi,n)$ has an optimal path in which every arc maximally increases $\ell$.
Let $G'$ be the graph where we drop all arcs from $G$ that do not maximally increase $\ell$.
Note that $G'$ still contains a path that corresponds to the global optimum.

The graph $G'$ has $O(k!\cdot n)$ nodes and each node has outdegree $O(k^3)$.
Then a breadth-first search from all start nodes to any node $(\pi^*,n)$ achieves the claimed time and space bounds, assuming we can enumerate the outgoing arcs of a node in constant time each.

For a given node $(\pi,\ell)$ we can enumerate all possible block
crossings in constant time each, as before.
In $G'$, we also need to know the maximum $\ell'$ such that all meetings $\ell+1$ up to $\ell'$ fit $\pi'$.
Note that $\ell'$ only depends on $\ell$ and $\pi'$.
We precompute a table $M(\pi,\ell)$ that gives this value.
Computing $M(\pi,\ell)$ for given $\pi$ and all $\ell$ takes a total of $O(kn)$ time: first compute for every $m_i$ whether it fits $\pi$, then compute the implied `forward pointers' using a linear scan.
So using $O(k!\cdot k\cdot n)$ preprocessing time and $O(k!\cdot n)$ space, we have an efficient implementation of the breadth-first search.
The theorem follows.
\end{proof}

\section{SBCM with Meetings of Two Characters}
\label{sec:algorithm}

\paragraph{A Greedy Algorithm.}

To quickly draw good storyline visualizations for 2-SBCM, we develop an
\mbox{$O(k n)$}-time greedy algorithm.
Given an instance $\CS=(C,M)$, we reserve a list $B = [\,]$ that the
algorithm will use to store the block crossings.
The algorithm starts with an arbitrary permutation~$\pi^0$ of $C$.
In every step the algorithm removes all meetings from the beginning
of~$M$ that fit the current permutation $\pi^i$ of the algorithm.
Subsequently, the algorithm picks a block crossing~$b$ such that
the resulting permutation~$\pi^{i+1}=b(\pi^i)$ supports the
maximum number of meetings from the beginning of~$M$.
Then $b$ is appended to the list~$B$.
This process repeats until $M$ is empty.
The algorithm returns %
$(\pi^0,B)$.

Note that there are at most $O(k^3)$ possible block crossings.
Thus to find the appropriate block crossings, the algorithm could simply check
all of them.
Many of those,
however, will result in permutations that do not even support the next
meeting, which would be a bad choice.
Hence, our algorithm considers only \emph{relevant} block crossings,
i.e., block crossings yielding a permutation that supports the next meeting.
Let $\{c,c'\}$ be the next meeting in $M$.
If $x$ and $y$ are the positions of~$c$ and~$c'$ in the current
permutation, i.e., $\pi^i_x = c$ and $\pi^i_y =c'$ (without loss of
generality, assume $x < y$), the relevant block crossings are:
\[ \{ (z, x, y-1) \colon 1 \leq z \leq x\} \cup
   \{ (x, z, y) \colon x \leq z < y \} \cup
   \{ (x+1, y-1, z) \colon y \leq z \leq k \}. \]
So the number of relevant block crossings in each step is $k+1$.
Let $n_i$ be the maximum number of meetings at the beginning of $M$ we can
achieve by one of these block crossings.
We use the data structure in Lemma~\ref{lem:datastruct} and check for each
relevant block crossing how many meetings can be done with this permutation.
Hence, we can identify a block crossing achieving the maximum number in $O(k n_i)$
time since we have to check $k+1$ paths containing up to $n_i$ meetings each.
Clearly, the numbers of meetings $n_i$ in each iteration of the algorithm sum
up to $n$ and therefore the algorithm runs in $O(k n)$ total time.

The way we described the greedy algorithm, it starts with an arbitrary
permutation.  Instead, we could start with a
permutation that supports the maximum number of meetings before the
the first block crossing needs to be done.
In other words, we want to find a maximal prefix $M'$ of $M$ such that
$(C, M')$ can be represented without any block
crossings.
We can find $M'$ in
$O(kn)$ time: we start with an empty graph and add the meetings
successively.  In each step we check whether the graph is still a
collection of paths, which can be done in $O(k)$ time.
It is easy to construct a permutation that supports all meetings in $M'$.
While this
is a sensible heuristic, we do not prove that this reduces the total number of
block crossings.
Indeed, we experimentally observe
that while the heuristic is generally good, this is not always the
case;
see Fig.~\ref{fig:greedy-fpt} for an example that uses the heuristic start permutation.

\begin{figure}[tb]
  \begin{subfigure}[b]{.52\textwidth}
    \centering
    \begin{tikzpicture}[yscale=.65,xscale=.435]
      \storylineset{
        norightlabel,
        leftwidth=0.2,
        rightwidth=0.3,
        distx=3,
        crossinglength=0.75
      }
      \drawstoryline[drawingstyle=char7,leftwidth=0.2,label=$7$]{1,1,4,7,6}
      \drawstoryline[drawingstyle=char5,leftwidth=0.2,label=$5$]{2,2,5,3,3}
      \drawstoryline[drawingstyle=char4,leftwidth=0.2,label=$4$]{3,3,6,4,7}
      \drawstoryline[drawingstyle=char2,leftwidth=0.2,label=$2$]{4,4,1,1,1}
      \drawstoryline[drawingstyle=char6,leftwidth=0.2,label=$6$]{5,5,2,2,2}
      \drawstoryline[drawingstyle=char3,leftwidth=0.2,label=$3$]{6,6,3,6,5}
      \drawstoryline[drawingstyle=char1,leftwidth=0.2,label=$1$]{7,7,7,5,4}

      \drawstorylinemeeting{0.0}{3}{4}
      \drawstorylinemeeting{0.2}{6}{7}
      \drawstorylinemeeting{0.4}{4}{5}
      \drawstorylinemeeting{0.6}{5}{6}
      \drawstorylinemeeting{0.8}{2}{3}
      \drawstorylinemeeting{1.0}{1}{2}

      \drawstorylinemeeting{1.82}{6}{7}
      \drawstorylinemeeting{2.0}{3}{4}

      \drawstorylinemeeting{2.92}{2}{3}

      \drawstorylinemeeting{3.85}{3}{4}
    \end{tikzpicture}

    \caption{Greedy solution}
    \label{fig:heuristic}
  \end{subfigure}
  \hfill
  \begin{subfigure}[b]{.42\textwidth}
    \centering
    \begin{tikzpicture}[yscale=.65,xscale=0.284]
      \storylineset{
        norightlabel,
        leftwidth=0.2,
        rightwidth=0.5,
        distx=3,
        crossinglength=0.75
      }
      \drawstoryline[drawingstyle=char7,leftwidth=0.2,label=$7$]{1,1,1,1,1,1}
      \drawstoryline[drawingstyle=char5,leftwidth=0.2,label=$5$]{2,2,2,2,4,4}
      \drawstoryline[drawingstyle=char4,leftwidth=0.2,label=$4$]{3,3,3,6,6,6}
      \drawstoryline[drawingstyle=char2,leftwidth=0.2,label=$2$]{4,4,4,7,7,7}
      \drawstoryline[drawingstyle=char6,leftwidth=0.2,label=$6$]{5,5,5,3,5,5}
      \drawstoryline[drawingstyle=char3,leftwidth=0.2,label=$3$]{6,6,6,4,2,2}
      \drawstoryline[drawingstyle=char1,leftwidth=0.2,label=$1$]{7,7,7,5,3,3}

      \drawstorylinemeeting{0.3}{3}{4}
      \drawstorylinemeeting{0.6}{6}{7}
      \drawstorylinemeeting{0.9}{4}{5}
      \drawstorylinemeeting{1.2}{5}{6}
      \drawstorylinemeeting{1.5}{2}{3}
      \drawstorylinemeeting{1.8}{1}{2}

      \drawstorylinemeeting{2.92}{5}{6}

      \drawstorylinemeeting{4.2}{1}{2}
      \drawstorylinemeeting{4.55}{4}{5}
      \drawstorylinemeeting{4.95}{3}{4}
    \end{tikzpicture}

    \caption{Optimal solution}
    \label{fig:fpt}
  \end{subfigure}
  \caption{The greedy algorithm is not optimal.}
  \label{fig:greedy-fpt}
\end{figure}
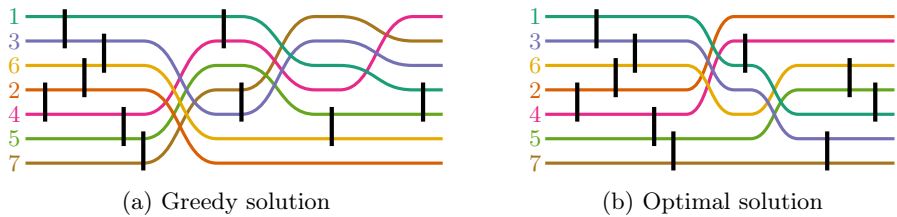

Note that the greedy algorithm yields optimal solutions for special cases of 2-SBCM.
The proof for the following theorem can be found in Appendix~\appref{app:2sbcm}.

\newcommand{\contentThmGreedyOptimal}{%
  For $k = 3$, the greedy algorithm produces optimal solutions.
}
\begin{theorem}
  \label{thm:k3-greedy-optimal}
  \contentThmGreedyOptimal
\end{theorem}

\paragraph{Experimental Evaluation.}

In this section, we report on some preliminary experimental results.
We only consider 2-SBCM.  We generated random instances as follows.
Given $n$ and $k$, we generate $n$ pairs of characters as meetings,
uniformly at random using rejection sampling to ensure that
consecutive meetings are different. (Repeated meetings are not
sensible.)

First, we consider the exact algorithm of
Theorem~\ref{thm:polySpaceExact}.  As expected, its runtime depends heavily
on~$k$ (Fig.~\ref{fig:greedy-all}, left).  Perhaps unexpectedly, we observe exponential runtime in $n$. This is actually a property of our random
instances, in which $\beta$ tends to increase linearly with $n$.
Note that this does not invalidate the algorithm since we may be
interested in instances for which $\beta$ is indeed small.

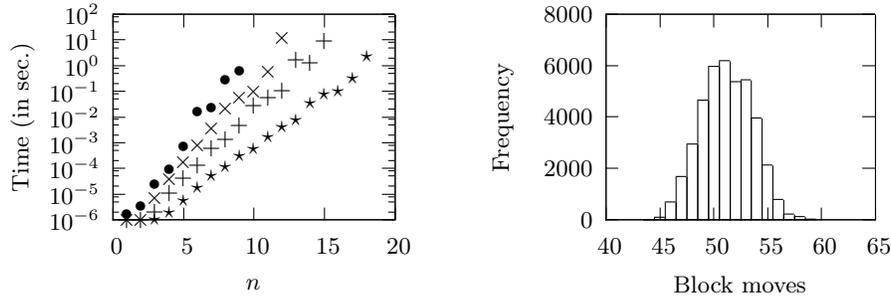
\begin{figure}[tb]
  \centering

  % GNUPLOT: LaTeX picture
\setlength{\unitlength}{0.240900pt}
\ifx\plotpoint\undefined\newsavebox{\plotpoint}\fi
\sbox{\plotpoint}{\rule[-0.200pt]{0.400pt}{0.400pt}}%
\begin{picture}(674,495)(0,0)
\sbox{\plotpoint}{\rule[-0.200pt]{0.400pt}{0.400pt}}%
\put(171.0,131.0){\rule[-0.200pt]{4.818pt}{0.400pt}}
\put(151,131){\makebox(0,0)[r]{$10^{-6}$}}
\put(594.0,131.0){\rule[-0.200pt]{4.818pt}{0.400pt}}
\put(171.0,143.0){\rule[-0.200pt]{2.409pt}{0.400pt}}
\put(604.0,143.0){\rule[-0.200pt]{2.409pt}{0.400pt}}
\put(171.0,159.0){\rule[-0.200pt]{2.409pt}{0.400pt}}
\put(604.0,159.0){\rule[-0.200pt]{2.409pt}{0.400pt}}
\put(171.0,167.0){\rule[-0.200pt]{2.409pt}{0.400pt}}
\put(604.0,167.0){\rule[-0.200pt]{2.409pt}{0.400pt}}
\put(171.0,171.0){\rule[-0.200pt]{4.818pt}{0.400pt}}
\put(151,171){\makebox(0,0)[r]{$10^{-5}$}}
\put(594.0,171.0){\rule[-0.200pt]{4.818pt}{0.400pt}}
\put(171.0,184.0){\rule[-0.200pt]{2.409pt}{0.400pt}}
\put(604.0,184.0){\rule[-0.200pt]{2.409pt}{0.400pt}}
\put(171.0,200.0){\rule[-0.200pt]{2.409pt}{0.400pt}}
\put(604.0,200.0){\rule[-0.200pt]{2.409pt}{0.400pt}}
\put(171.0,208.0){\rule[-0.200pt]{2.409pt}{0.400pt}}
\put(604.0,208.0){\rule[-0.200pt]{2.409pt}{0.400pt}}
\put(171.0,212.0){\rule[-0.200pt]{4.818pt}{0.400pt}}
\put(151,212){\makebox(0,0)[r]{$10^{-4}$}}
\put(594.0,212.0){\rule[-0.200pt]{4.818pt}{0.400pt}}
\put(171.0,224.0){\rule[-0.200pt]{2.409pt}{0.400pt}}
\put(604.0,224.0){\rule[-0.200pt]{2.409pt}{0.400pt}}
\put(171.0,240.0){\rule[-0.200pt]{2.409pt}{0.400pt}}
\put(604.0,240.0){\rule[-0.200pt]{2.409pt}{0.400pt}}
\put(171.0,248.0){\rule[-0.200pt]{2.409pt}{0.400pt}}
\put(604.0,248.0){\rule[-0.200pt]{2.409pt}{0.400pt}}
\put(171.0,252.0){\rule[-0.200pt]{4.818pt}{0.400pt}}
\put(151,252){\makebox(0,0)[r]{$10^{-3}$}}
\put(594.0,252.0){\rule[-0.200pt]{4.818pt}{0.400pt}}
\put(171.0,264.0){\rule[-0.200pt]{2.409pt}{0.400pt}}
\put(604.0,264.0){\rule[-0.200pt]{2.409pt}{0.400pt}}
\put(171.0,280.0){\rule[-0.200pt]{2.409pt}{0.400pt}}
\put(604.0,280.0){\rule[-0.200pt]{2.409pt}{0.400pt}}
\put(171.0,289.0){\rule[-0.200pt]{2.409pt}{0.400pt}}
\put(604.0,289.0){\rule[-0.200pt]{2.409pt}{0.400pt}}
\put(171.0,292.0){\rule[-0.200pt]{4.818pt}{0.400pt}}
\put(151,292){\makebox(0,0)[r]{$10^{-2}$}}
\put(594.0,292.0){\rule[-0.200pt]{4.818pt}{0.400pt}}
\put(171.0,305.0){\rule[-0.200pt]{2.409pt}{0.400pt}}
\put(604.0,305.0){\rule[-0.200pt]{2.409pt}{0.400pt}}
\put(171.0,321.0){\rule[-0.200pt]{2.409pt}{0.400pt}}
\put(604.0,321.0){\rule[-0.200pt]{2.409pt}{0.400pt}}
\put(171.0,329.0){\rule[-0.200pt]{2.409pt}{0.400pt}}
\put(604.0,329.0){\rule[-0.200pt]{2.409pt}{0.400pt}}
\put(171.0,333.0){\rule[-0.200pt]{4.818pt}{0.400pt}}
\put(151,333){\makebox(0,0)[r]{$10^{-1}$}}
\put(594.0,333.0){\rule[-0.200pt]{4.818pt}{0.400pt}}
\put(171.0,345.0){\rule[-0.200pt]{2.409pt}{0.400pt}}
\put(604.0,345.0){\rule[-0.200pt]{2.409pt}{0.400pt}}
\put(171.0,361.0){\rule[-0.200pt]{2.409pt}{0.400pt}}
\put(604.0,361.0){\rule[-0.200pt]{2.409pt}{0.400pt}}
\put(171.0,369.0){\rule[-0.200pt]{2.409pt}{0.400pt}}
\put(604.0,369.0){\rule[-0.200pt]{2.409pt}{0.400pt}}
\put(171.0,373.0){\rule[-0.200pt]{4.818pt}{0.400pt}}
\put(151,373){\makebox(0,0)[r]{$10^{0}$}}
\put(594.0,373.0){\rule[-0.200pt]{4.818pt}{0.400pt}}
\put(171.0,385.0){\rule[-0.200pt]{2.409pt}{0.400pt}}
\put(604.0,385.0){\rule[-0.200pt]{2.409pt}{0.400pt}}
\put(171.0,401.0){\rule[-0.200pt]{2.409pt}{0.400pt}}
\put(604.0,401.0){\rule[-0.200pt]{2.409pt}{0.400pt}}
\put(171.0,410.0){\rule[-0.200pt]{2.409pt}{0.400pt}}
\put(604.0,410.0){\rule[-0.200pt]{2.409pt}{0.400pt}}
\put(171.0,414.0){\rule[-0.200pt]{4.818pt}{0.400pt}}
\put(151,414){\makebox(0,0)[r]{$10^{1}$}}
\put(594.0,414.0){\rule[-0.200pt]{4.818pt}{0.400pt}}
\put(171.0,426.0){\rule[-0.200pt]{2.409pt}{0.400pt}}
\put(604.0,426.0){\rule[-0.200pt]{2.409pt}{0.400pt}}
\put(171.0,442.0){\rule[-0.200pt]{2.409pt}{0.400pt}}
\put(604.0,442.0){\rule[-0.200pt]{2.409pt}{0.400pt}}
\put(171.0,450.0){\rule[-0.200pt]{2.409pt}{0.400pt}}
\put(604.0,450.0){\rule[-0.200pt]{2.409pt}{0.400pt}}
\put(171.0,454.0){\rule[-0.200pt]{4.818pt}{0.400pt}}
\put(151,454){\makebox(0,0)[r]{$10^{2}$}}
\put(594.0,454.0){\rule[-0.200pt]{4.818pt}{0.400pt}}
\put(171.0,131.0){\rule[-0.200pt]{0.400pt}{4.818pt}}
\put(171,90){\makebox(0,0){ 0}}
\put(171.0,434.0){\rule[-0.200pt]{0.400pt}{4.818pt}}
\put(282.0,131.0){\rule[-0.200pt]{0.400pt}{4.818pt}}
\put(282,90){\makebox(0,0){ 5}}
\put(282.0,434.0){\rule[-0.200pt]{0.400pt}{4.818pt}}
\put(392.0,131.0){\rule[-0.200pt]{0.400pt}{4.818pt}}
\put(392,90){\makebox(0,0){ 10}}
\put(392.0,434.0){\rule[-0.200pt]{0.400pt}{4.818pt}}
\put(503.0,131.0){\rule[-0.200pt]{0.400pt}{4.818pt}}
\put(503,90){\makebox(0,0){ 15}}
\put(503.0,434.0){\rule[-0.200pt]{0.400pt}{4.818pt}}
\put(614.0,131.0){\rule[-0.200pt]{0.400pt}{4.818pt}}
\put(614,90){\makebox(0,0){ 20}}
\put(614.0,434.0){\rule[-0.200pt]{0.400pt}{4.818pt}}
\put(171.0,131.0){\rule[-0.200pt]{0.400pt}{77.811pt}}
\put(171.0,131.0){\rule[-0.200pt]{106.719pt}{0.400pt}}
\put(614.0,131.0){\rule[-0.200pt]{0.400pt}{77.811pt}}
\put(171.0,454.0){\rule[-0.200pt]{106.719pt}{0.400pt}}
\put(30,292){\makebox(0,0){\rotatebox{90}{Time (in sec.)}}}
\put(392,29){\makebox(0,0){$n$}}
\put(237,131){\makebox(0,0){$\star$}}
\put(260,143){\makebox(0,0){$\star$}}
\put(282,162){\makebox(0,0){$\star$}}
\put(304,182){\makebox(0,0){$\star$}}
\put(326,200){\makebox(0,0){$\star$}}
\put(348,215){\makebox(0,0){$\star$}}
\put(370,232){\makebox(0,0){$\star$}}
\put(392,243){\makebox(0,0){$\star$}}
\put(415,261){\makebox(0,0){$\star$}}
\put(437,277){\makebox(0,0){$\star$}}
\put(459,288){\makebox(0,0){$\star$}}
\put(481,314){\makebox(0,0){$\star$}}
\put(503,329){\makebox(0,0){$\star$}}
\put(525,333){\makebox(0,0){$\star$}}
\put(548,353){\makebox(0,0){$\star$}}
\put(570,388){\makebox(0,0){$\star$}}
\put(193,131){\makebox(0,0){$+$}}
\put(215,131){\makebox(0,0){$+$}}
\put(237,143){\makebox(0,0){$+$}}
\put(260,173){\makebox(0,0){$+$}}
\put(282,196){\makebox(0,0){$+$}}
\put(304,216){\makebox(0,0){$+$}}
\put(326,243){\makebox(0,0){$+$}}
\put(348,257){\makebox(0,0){$+$}}
\put(370,279){\makebox(0,0){$+$}}
\put(392,310){\makebox(0,0){$+$}}
\put(415,323){\makebox(0,0){$+$}}
\put(437,334){\makebox(0,0){$+$}}
\put(459,382){\makebox(0,0){$+$}}
\put(481,378){\makebox(0,0){$+$}}
\put(503,412){\makebox(0,0){$+$}}
\sbox{\plotpoint}{\rule[-0.400pt]{0.800pt}{0.800pt}}%
\sbox{\plotpoint}{\rule[-0.200pt]{0.400pt}{0.400pt}}%
\put(193,131){\makebox(0,0){$\times$}}
\put(215,131){\makebox(0,0){$\times$}}
\put(237,165){\makebox(0,0){$\times$}}
\put(260,195){\makebox(0,0){$\times$}}
\put(282,221){\makebox(0,0){$\times$}}
\put(304,248){\makebox(0,0){$\times$}}
\put(326,274){\makebox(0,0){$\times$}}
\put(348,305){\makebox(0,0){$\times$}}
\put(370,323){\makebox(0,0){$\times$}}
\put(392,333){\makebox(0,0){$\times$}}
\put(415,364){\makebox(0,0){$\times$}}
\put(437,417){\makebox(0,0){$\times$}}
\sbox{\plotpoint}{\rule[-0.500pt]{1.000pt}{1.000pt}}%
\sbox{\plotpoint}{\rule[-0.200pt]{0.400pt}{0.400pt}}%
\put(193,143){\raisebox{-.8pt}{\makebox(0,0){$\bullet$}}}
\put(215,155){\raisebox{-.8pt}{\makebox(0,0){$\bullet$}}}
\put(237,189){\raisebox{-.8pt}{\makebox(0,0){$\bullet$}}}
\put(260,213){\raisebox{-.8pt}{\makebox(0,0){$\bullet$}}}
\put(282,249){\raisebox{-.8pt}{\makebox(0,0){$\bullet$}}}
\put(304,303){\raisebox{-.8pt}{\makebox(0,0){$\bullet$}}}
\put(326,309){\raisebox{-.8pt}{\makebox(0,0){$\bullet$}}}
\put(348,353){\raisebox{-.8pt}{\makebox(0,0){$\bullet$}}}
\put(370,368){\raisebox{-.8pt}{\makebox(0,0){$\bullet$}}}
\put(171.0,131.0){\rule[-0.200pt]{0.400pt}{77.811pt}}
\put(171.0,131.0){\rule[-0.200pt]{106.719pt}{0.400pt}}
\put(614.0,131.0){\rule[-0.200pt]{0.400pt}{77.811pt}}
\put(171.0,454.0){\rule[-0.200pt]{106.719pt}{0.400pt}}
\end{picture} \hfill
  % GNUPLOT: LaTeX picture
\setlength{\unitlength}{0.240900pt}
\ifx\plotpoint\undefined\newsavebox{\plotpoint}\fi
\sbox{\plotpoint}{\rule[-0.200pt]{0.400pt}{0.400pt}}%
\begin{picture}(674,495)(0,0)
\sbox{\plotpoint}{\rule[-0.200pt]{0.400pt}{0.400pt}}%
\put(191.0,131.0){\rule[-0.200pt]{4.818pt}{0.400pt}}
\put(171,131){\makebox(0,0)[r]{ 0}}
\put(594.0,131.0){\rule[-0.200pt]{4.818pt}{0.400pt}}
\put(191.0,212.0){\rule[-0.200pt]{4.818pt}{0.400pt}}
\put(171,212){\makebox(0,0)[r]{ 2000}}
\put(594.0,212.0){\rule[-0.200pt]{4.818pt}{0.400pt}}
\put(191.0,293.0){\rule[-0.200pt]{4.818pt}{0.400pt}}
\put(171,293){\makebox(0,0)[r]{ 4000}}
\put(594.0,293.0){\rule[-0.200pt]{4.818pt}{0.400pt}}
\put(191.0,373.0){\rule[-0.200pt]{4.818pt}{0.400pt}}
\put(171,373){\makebox(0,0)[r]{ 6000}}
\put(594.0,373.0){\rule[-0.200pt]{4.818pt}{0.400pt}}
\put(191.0,454.0){\rule[-0.200pt]{4.818pt}{0.400pt}}
\put(171,454){\makebox(0,0)[r]{ 8000}}
\put(594.0,454.0){\rule[-0.200pt]{4.818pt}{0.400pt}}
\put(191.0,131.0){\rule[-0.200pt]{0.400pt}{4.818pt}}
\put(191,90){\makebox(0,0){ 40}}
\put(191.0,434.0){\rule[-0.200pt]{0.400pt}{4.818pt}}
\put(276.0,131.0){\rule[-0.200pt]{0.400pt}{4.818pt}}
\put(276,90){\makebox(0,0){ 45}}
\put(276.0,434.0){\rule[-0.200pt]{0.400pt}{4.818pt}}
\put(360.0,131.0){\rule[-0.200pt]{0.400pt}{4.818pt}}
\put(360,90){\makebox(0,0){ 50}}
\put(360.0,434.0){\rule[-0.200pt]{0.400pt}{4.818pt}}
\put(445.0,131.0){\rule[-0.200pt]{0.400pt}{4.818pt}}
\put(445,90){\makebox(0,0){ 55}}
\put(445.0,434.0){\rule[-0.200pt]{0.400pt}{4.818pt}}
\put(529.0,131.0){\rule[-0.200pt]{0.400pt}{4.818pt}}
\put(529,90){\makebox(0,0){ 60}}
\put(529.0,434.0){\rule[-0.200pt]{0.400pt}{4.818pt}}
\put(614.0,131.0){\rule[-0.200pt]{0.400pt}{4.818pt}}
\put(614,90){\makebox(0,0){ 65}}
\put(614.0,434.0){\rule[-0.200pt]{0.400pt}{4.818pt}}
\put(191.0,131.0){\rule[-0.200pt]{0.400pt}{77.811pt}}
\put(191.0,131.0){\rule[-0.200pt]{101.901pt}{0.400pt}}
\put(614.0,131.0){\rule[-0.200pt]{0.400pt}{77.811pt}}
\put(191.0,454.0){\rule[-0.200pt]{101.901pt}{0.400pt}}
\put(30,292){\makebox(0,0){\rotatebox{90}{Frequency}}}
\put(402,29){\makebox(0,0){Block moves}}
\put(250,131){\usebox{\plotpoint}}
\put(250.0,131.0){\rule[-0.200pt]{4.095pt}{0.400pt}}
\put(250.0,131.0){\rule[-0.200pt]{4.095pt}{0.400pt}}
\put(267.0,131.0){\rule[-0.200pt]{0.400pt}{0.964pt}}
\put(267.0,135.0){\rule[-0.200pt]{4.095pt}{0.400pt}}
\put(284.0,131.0){\rule[-0.200pt]{0.400pt}{0.964pt}}
\put(267.0,131.0){\rule[-0.200pt]{4.095pt}{0.400pt}}
\put(284.0,131.0){\rule[-0.200pt]{0.400pt}{6.745pt}}
\put(284.0,159.0){\rule[-0.200pt]{4.095pt}{0.400pt}}
\put(301.0,131.0){\rule[-0.200pt]{0.400pt}{6.745pt}}
\put(284.0,131.0){\rule[-0.200pt]{4.095pt}{0.400pt}}
\put(301.0,131.0){\rule[-0.200pt]{0.400pt}{16.381pt}}
\put(301.0,199.0){\rule[-0.200pt]{4.095pt}{0.400pt}}
\put(318.0,131.0){\rule[-0.200pt]{0.400pt}{16.381pt}}
\put(301.0,131.0){\rule[-0.200pt]{4.095pt}{0.400pt}}
\put(318.0,131.0){\rule[-0.200pt]{0.400pt}{28.667pt}}
\put(318.0,250.0){\rule[-0.200pt]{4.095pt}{0.400pt}}
\put(335.0,131.0){\rule[-0.200pt]{0.400pt}{28.667pt}}
\put(318.0,131.0){\rule[-0.200pt]{4.095pt}{0.400pt}}
\put(335.0,131.0){\rule[-0.200pt]{0.400pt}{45.289pt}}
\put(335.0,319.0){\rule[-0.200pt]{4.095pt}{0.400pt}}
\put(352.0,131.0){\rule[-0.200pt]{0.400pt}{45.289pt}}
\put(335.0,131.0){\rule[-0.200pt]{4.095pt}{0.400pt}}
\put(352.0,131.0){\rule[-0.200pt]{0.400pt}{58.057pt}}
\put(352.0,372.0){\rule[-0.200pt]{4.095pt}{0.400pt}}
\put(369.0,131.0){\rule[-0.200pt]{0.400pt}{58.057pt}}
\put(352.0,131.0){\rule[-0.200pt]{4.095pt}{0.400pt}}
\put(369.0,131.0){\rule[-0.200pt]{0.400pt}{60.225pt}}
\put(369.0,381.0){\rule[-0.200pt]{4.095pt}{0.400pt}}
\put(386.0,131.0){\rule[-0.200pt]{0.400pt}{60.225pt}}
\put(369.0,131.0){\rule[-0.200pt]{4.095pt}{0.400pt}}
\put(386.0,131.0){\rule[-0.200pt]{0.400pt}{52.275pt}}
\put(386.0,348.0){\rule[-0.200pt]{4.095pt}{0.400pt}}
\put(403.0,131.0){\rule[-0.200pt]{0.400pt}{52.275pt}}
\put(386.0,131.0){\rule[-0.200pt]{4.095pt}{0.400pt}}
\put(403.0,131.0){\rule[-0.200pt]{0.400pt}{52.998pt}}
\put(403.0,351.0){\rule[-0.200pt]{3.854pt}{0.400pt}}
\put(419.0,131.0){\rule[-0.200pt]{0.400pt}{52.998pt}}
\put(403.0,131.0){\rule[-0.200pt]{3.854pt}{0.400pt}}
\put(419.0,131.0){\rule[-0.200pt]{0.400pt}{38.544pt}}
\put(419.0,291.0){\rule[-0.200pt]{4.095pt}{0.400pt}}
\put(436.0,131.0){\rule[-0.200pt]{0.400pt}{38.544pt}}
\put(419.0,131.0){\rule[-0.200pt]{4.095pt}{0.400pt}}
\put(436.0,131.0){\rule[-0.200pt]{0.400pt}{20.717pt}}
\put(436.0,217.0){\rule[-0.200pt]{4.095pt}{0.400pt}}
\put(453.0,131.0){\rule[-0.200pt]{0.400pt}{20.717pt}}
\put(436.0,131.0){\rule[-0.200pt]{4.095pt}{0.400pt}}
\put(453.0,131.0){\rule[-0.200pt]{0.400pt}{7.709pt}}
\put(453.0,163.0){\rule[-0.200pt]{4.095pt}{0.400pt}}
\put(470.0,131.0){\rule[-0.200pt]{0.400pt}{7.709pt}}
\put(453.0,131.0){\rule[-0.200pt]{4.095pt}{0.400pt}}
\put(470.0,131.0){\rule[-0.200pt]{0.400pt}{2.168pt}}
\put(470.0,140.0){\rule[-0.200pt]{4.095pt}{0.400pt}}
\put(487.0,131.0){\rule[-0.200pt]{0.400pt}{2.168pt}}
\put(470.0,131.0){\rule[-0.200pt]{4.095pt}{0.400pt}}
\put(487.0,131.0){\rule[-0.200pt]{0.400pt}{1.204pt}}
\put(487.0,136.0){\rule[-0.200pt]{4.095pt}{0.400pt}}
\put(504.0,131.0){\rule[-0.200pt]{0.400pt}{1.204pt}}
\put(487.0,131.0){\rule[-0.200pt]{4.095pt}{0.400pt}}
\put(504.0,131.0){\usebox{\plotpoint}}
\put(504.0,132.0){\rule[-0.200pt]{4.095pt}{0.400pt}}
\put(521.0,131.0){\usebox{\plotpoint}}
\put(504.0,131.0){\rule[-0.200pt]{4.095pt}{0.400pt}}
\put(191.0,131.0){\rule[-0.200pt]{0.400pt}{77.811pt}}
\put(191.0,131.0){\rule[-0.200pt]{101.901pt}{0.400pt}}
\put(614.0,131.0){\rule[-0.200pt]{0.400pt}{77.811pt}}
\put(191.0,454.0){\rule[-0.200pt]{101.901pt}{0.400pt}}
\end{picture}

  \caption{Left: Runtime of the exact algorithm of
    Theorem~\ref{thm:polySpaceExact} on random instances with $k=4
    (\star), 5 (+),6 (\times),7 (\bullet)$.  Each data point is the
    average of $50$ random instances.  Right: Histogram of the number of
    block crossings used by the greedy algorithm for all $k!$
    different start permutations, on a single random instance with
    $n=100$ and $k=8$.}%
  \label{fig:greedy-all}
\end{figure}

Since the exact algorithm is feasible only for rather small instances,
we now shift our focus to the greedy algorithm.
Recall that it starts with an
arbitrary permutation and proceeds greedily.
The histogram in Fig.~\ref{fig:greedy-all} (right) shows the
number of block
crossings used by the greedy algorithm depending on the start permutation, for a single random
instance:  this bell curve is typical.  We see that there are
``rare'' start permutations that do strictly better than almost all others.  Indeed,
for the reported instance, a random start permutation does $7.2$ block
crossings worse in expectation than the best possible start permutation.

\begin{figure}[tb]
  \centering
  % GNUPLOT: LaTeX picture
\setlength{\unitlength}{0.240900pt}
\ifx\plotpoint\undefined\newsavebox{\plotpoint}\fi
\sbox{\plotpoint}{\rule[-0.200pt]{0.400pt}{0.400pt}}%
\begin{picture}(674,360)(0,0)
\sbox{\plotpoint}{\rule[-0.200pt]{0.400pt}{0.400pt}}%
\put(151.0,131.0){\rule[-0.200pt]{4.818pt}{0.400pt}}
\put(131,131){\makebox(0,0)[r]{ 0}}
\put(594.0,131.0){\rule[-0.200pt]{4.818pt}{0.400pt}}
\put(151.0,178.0){\rule[-0.200pt]{4.818pt}{0.400pt}}
\put(131,178){\makebox(0,0)[r]{ 10}}
\put(594.0,178.0){\rule[-0.200pt]{4.818pt}{0.400pt}}
\put(151.0,225.0){\rule[-0.200pt]{4.818pt}{0.400pt}}
\put(131,225){\makebox(0,0)[r]{ 20}}
\put(594.0,225.0){\rule[-0.200pt]{4.818pt}{0.400pt}}
\put(151.0,273.0){\rule[-0.200pt]{4.818pt}{0.400pt}}
\put(131,273){\makebox(0,0)[r]{ 30}}
\put(594.0,273.0){\rule[-0.200pt]{4.818pt}{0.400pt}}
\put(151.0,320.0){\rule[-0.200pt]{4.818pt}{0.400pt}}
\put(131,320){\makebox(0,0)[r]{ 40}}
\put(594.0,320.0){\rule[-0.200pt]{4.818pt}{0.400pt}}
\put(151.0,131.0){\rule[-0.200pt]{0.400pt}{4.818pt}}
\put(151,90){\makebox(0,0){-5}}
\put(151.0,300.0){\rule[-0.200pt]{0.400pt}{4.818pt}}
\put(244.0,131.0){\rule[-0.200pt]{0.400pt}{4.818pt}}
\put(244,90){\makebox(0,0){ 0}}
\put(244.0,300.0){\rule[-0.200pt]{0.400pt}{4.818pt}}
\put(336.0,131.0){\rule[-0.200pt]{0.400pt}{4.818pt}}
\put(336,90){\makebox(0,0){ 5}}
\put(336.0,300.0){\rule[-0.200pt]{0.400pt}{4.818pt}}
\put(429.0,131.0){\rule[-0.200pt]{0.400pt}{4.818pt}}
\put(429,90){\makebox(0,0){ 10}}
\put(429.0,300.0){\rule[-0.200pt]{0.400pt}{4.818pt}}
\put(521.0,131.0){\rule[-0.200pt]{0.400pt}{4.818pt}}
\put(521,90){\makebox(0,0){ 15}}
\put(521.0,300.0){\rule[-0.200pt]{0.400pt}{4.818pt}}
\put(614.0,131.0){\rule[-0.200pt]{0.400pt}{4.818pt}}
\put(614,90){\makebox(0,0){ 20}}
\put(614.0,300.0){\rule[-0.200pt]{0.400pt}{4.818pt}}
\put(151.0,131.0){\rule[-0.200pt]{0.400pt}{45.530pt}}
\put(151.0,131.0){\rule[-0.200pt]{111.537pt}{0.400pt}}
\put(614.0,131.0){\rule[-0.200pt]{0.400pt}{45.530pt}}
\put(151.0,320.0){\rule[-0.200pt]{111.537pt}{0.400pt}}
\put(30,225){\makebox(0,0){\rotatebox{90}{Frequency}}}
\put(382,29){\makebox(0,0){Excess block moves}}
\put(234.0,131.0){\rule[-0.200pt]{0.400pt}{14.695pt}}
\put(234.0,192.0){\rule[-0.200pt]{4.577pt}{0.400pt}}
\put(253.0,131.0){\rule[-0.200pt]{0.400pt}{14.695pt}}
\put(234.0,131.0){\rule[-0.200pt]{4.577pt}{0.400pt}}
\put(253.0,131.0){\rule[-0.200pt]{0.400pt}{20.476pt}}
\put(253.0,216.0){\rule[-0.200pt]{4.336pt}{0.400pt}}
\put(271.0,131.0){\rule[-0.200pt]{0.400pt}{20.476pt}}
\put(253.0,131.0){\rule[-0.200pt]{4.336pt}{0.400pt}}
\put(271.0,131.0){\rule[-0.200pt]{0.400pt}{28.426pt}}
\put(271.0,249.0){\rule[-0.200pt]{4.577pt}{0.400pt}}
\put(290.0,131.0){\rule[-0.200pt]{0.400pt}{28.426pt}}
\put(271.0,131.0){\rule[-0.200pt]{4.577pt}{0.400pt}}
\put(290.0,131.0){\rule[-0.200pt]{0.400pt}{38.785pt}}
\put(290.0,292.0){\rule[-0.200pt]{4.336pt}{0.400pt}}
\put(308.0,131.0){\rule[-0.200pt]{0.400pt}{38.785pt}}
\put(290.0,131.0){\rule[-0.200pt]{4.336pt}{0.400pt}}
\put(308.0,131.0){\rule[-0.200pt]{0.400pt}{40.953pt}}
\put(308.0,301.0){\rule[-0.200pt]{4.577pt}{0.400pt}}
\put(327.0,131.0){\rule[-0.200pt]{0.400pt}{40.953pt}}
\put(308.0,131.0){\rule[-0.200pt]{4.577pt}{0.400pt}}
\put(327.0,131.0){\rule[-0.200pt]{0.400pt}{25.054pt}}
\put(327.0,235.0){\rule[-0.200pt]{4.336pt}{0.400pt}}
\put(345.0,131.0){\rule[-0.200pt]{0.400pt}{25.054pt}}
\put(327.0,131.0){\rule[-0.200pt]{4.336pt}{0.400pt}}
\put(345.0,131.0){\rule[-0.200pt]{0.400pt}{14.695pt}}
\put(345.0,192.0){\rule[-0.200pt]{4.577pt}{0.400pt}}
\put(364.0,131.0){\rule[-0.200pt]{0.400pt}{14.695pt}}
\put(345.0,131.0){\rule[-0.200pt]{4.577pt}{0.400pt}}
\put(364.0,131.0){\rule[-0.200pt]{0.400pt}{22.645pt}}
\put(364.0,225.0){\rule[-0.200pt]{4.336pt}{0.400pt}}
\put(382.0,131.0){\rule[-0.200pt]{0.400pt}{22.645pt}}
\put(364.0,131.0){\rule[-0.200pt]{4.336pt}{0.400pt}}
\put(382.0,131.0){\rule[-0.200pt]{0.400pt}{11.322pt}}
\put(382.0,178.0){\rule[-0.200pt]{4.577pt}{0.400pt}}
\put(401.0,131.0){\rule[-0.200pt]{0.400pt}{11.322pt}}
\put(382.0,131.0){\rule[-0.200pt]{4.577pt}{0.400pt}}
\put(401.0,131.0){\rule[-0.200pt]{0.400pt}{3.373pt}}
\put(401.0,145.0){\rule[-0.200pt]{4.577pt}{0.400pt}}
\put(420.0,131.0){\rule[-0.200pt]{0.400pt}{3.373pt}}
\put(401.0,131.0){\rule[-0.200pt]{4.577pt}{0.400pt}}
\put(420.0,131.0){\rule[-0.200pt]{0.400pt}{3.373pt}}
\put(420.0,145.0){\rule[-0.200pt]{4.336pt}{0.400pt}}
\put(438.0,131.0){\rule[-0.200pt]{0.400pt}{3.373pt}}
\put(420.0,131.0){\rule[-0.200pt]{4.336pt}{0.400pt}}
\put(457.0,131.0){\rule[-0.200pt]{0.400pt}{1.204pt}}
\put(457.0,136.0){\rule[-0.200pt]{4.336pt}{0.400pt}}
\put(475.0,131.0){\rule[-0.200pt]{0.400pt}{1.204pt}}
\put(457.0,131.0){\rule[-0.200pt]{4.336pt}{0.400pt}}
\put(475.0,131.0){\rule[-0.200pt]{0.400pt}{2.168pt}}
\put(475.0,140.0){\rule[-0.200pt]{4.577pt}{0.400pt}}
\put(494.0,131.0){\rule[-0.200pt]{0.400pt}{2.168pt}}
\put(475.0,131.0){\rule[-0.200pt]{4.577pt}{0.400pt}}
\put(151.0,131.0){\rule[-0.200pt]{0.400pt}{45.530pt}}
\put(151.0,131.0){\rule[-0.200pt]{111.537pt}{0.400pt}}
\put(614.0,131.0){\rule[-0.200pt]{0.400pt}{45.530pt}}
\put(151.0,320.0){\rule[-0.200pt]{111.537pt}{0.400pt}}
\end{picture}
  \hfill \input{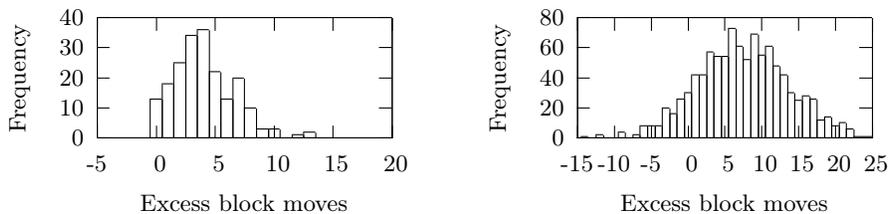}

  \caption{Left: histogram of \algo{HeuristicGreedy} minus \algo{BestGreedy}, $200$ instances with
    with $k=7$ and $n=100$.
    Right: histogram of \algo{RandomGreedy} minus \algo{HeuristicGreedy}, $1000$ instances with $k=30$ and
    $n=200$.}
  \label{fig:random-heuristic}
\end{figure}

We call the best possible result of the greedy algorithm over all start permutations \algo{BestGreedy}, which we calculate by brute force.
Let \algo{RandomGreedy} start with a permutation chosen uniformly at
random, and let \algo{HeuristicGreedy} start with the heuristic
start permutation that we have described above.
The histogram in Fig.~\ref{fig:random-heuristic} (left) shows
how many more block crossings \algo{HeuristicGreedy} uses than \algo{BestGreedy} on random instances.
This distribution is heaviest near zero, but there are instances where performance is poor.
Note that we do not know how to compute \algo{BestGreedy} efficiently.
Compared to \algo{RandomGreedy}, we see that \algo{HeuristicGreedy}
fares well (Fig.~\ref{fig:random-heuristic}, right).

Lastly, we compare the greedy algorithm to the
optimum, which we can only do for small $k$ and $n$.
On $1000$ random instances with
$k=5$ and $n=12$, \algo{HeuristicGreedy}
was optimal $56\%$ of the time.
It was sometimes off by one ($38\%$), two ($5\%$), or three ($1\%$), but never worse.
This is a promising behavior, but clearly cannot be extrapolated
verbatim to larger instances. 

Based on these experiments, we recommend \algo{HeuristicGreedy} as an efficient, reasonable heuristic.

\section{Approximation Algorithm}
\label{sec:approximation}
We now develop a constant-factor approximation algorithm for
$d$-SBCM where $d$ is a constant.
We initially assume that each group meeting occurs exactly once, but later
show how to extend our results to the setting where the same group can
meet a bounded number of times.

\paragraph{Overview.}
Our approximation algorithm has the following three main steps.
\begin{enumerate} \advance\itemsep by 2pt
\item Reduce the input group hypergraph $\hyper=(C,\hedges)$ to an
  \emph{interval hypergraph} $\hypersp = (C,~ \hedges \setminus
  \hedgespp)$ by deleting a subset $\hedgespp \subseteq
  \hedges$ of the edges of \hyper. 

\item Choose a permutation $\pi^0$ of the characters that supports all
  groups of this interval hypergraph $\hypersp$. Thus, $\pi^0$ is the
  order of characters at the beginning of the timeline.

\item Incrementally create support for each deleted meeting of
  $\hedgespp$ in order of increasing time, as follows. Suppose that $g
  \in \hedgespp$ is the group meeting to support. Keep one of the
  character lines involved in this meeting fixed and bring, for the
  duration of the meeting, the remaining (at most $d-1$) lines close
  to it.  Then retract those lines to their original position
  in~$\pi^0$; see Fig.~\ref{fig:bc-alg-example}.
\end{enumerate}

\begin{wrapfigure}[7]{r}{.27\textwidth}
  \centering
  \includegraphics{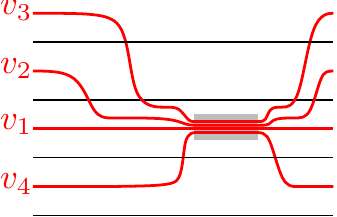}

  \parbox{.7\linewidth}{\caption{meeting $\{v_1, v_2, v_3, v_4 \}$
      \label{fig:bc-alg-example}}}
\end{wrapfigure}
Step 2 is straightforward: Section~\ref{sec:preliminaries} shows how
to find a permutation supporting all the groups for an interval hypergraph.
In Step~3, we introduce at most $2(d-1)$ block crossings for each
meeting 
$g \in \hedgespp$ not initially supported. The main technical 
parts of the algorithm are Step 1 and an analysis to charge at most a 
constant number of block crossings in Step 3 to a block crossing in the 
optimal visualization.
Step 1 requires solving a hypergraph problem; this is technically the
most challenging part, and consumes the entire
Section~\ref{sec:hyperedge-removal-approx}.

\paragraph{Bounds and Analysis.}
We call $\hedgespp$ \emph{paid} edges, and the remainder $\hedgesp
\;=\; \hedges \setminus \hedgespp$ 
\emph{free} edges. Intuitively, free edges can be realized without block 
crossings because $\hypersp$ is an interval hypergraph, while the edges 
of $\hedgespp$ must be charged to block crossings of the optimal
drawing.  We initialize the drawing by placing the characters in the
vertical order $\pi^0$, which supports all the groups in~$\hedgesp$.
Now we consider the paid edges in left-to-right order.
Suppose that the next meeting involves a group $g'
\in \hedgespp$. We have $|g'| \le d$.
We arbitrarily fix one of its characters, leaving its line intact,
and bring the remaining $(d-1)$ lines in its vicinity to realize the
meeting. This creates at most $(d-1)$ block crossings, one per line.
When the meeting is over, we again use up to $(d-1)$ block
crossings to revert the lines back to their original position
prescribed by~$\pi^0$; see Fig.~\ref{fig:bc-alg-example}.

We do this for each paid hyperedge, giving rise to at most
$2(d-1)|\hedgespp|$ block 
crossings. We now prove that this bound is within a constant factor of optimal.
We first establish a lower bound on the optimal number of block crossings
\emph{assuming} that $\pi^0$ is the optimal start permutation.

\begin{lemma}
  \label{lem:hyper-story-lower-bound}
  Let $\pi$ be a permutation of the characters, let $\hedgesp$ be the
  groups supported by~$\pi$, and let $\hedgespp \;=\; \hedges
  \setminus \hedgesp$.  Any storyline visualization that uses $\pi$ as
  the start permutation has at least $4|\hedgespp|/(3d^2)$ block
  crossings.
\end{lemma}
\begin{proof}
  Let $g \in \hedgespp$. Since $g$ is not supported by
  $\pi$, the optimal drawing does not contain the
  characters of $g$ as a contiguous block initially. However, in order to
  support this meeting, these characters must eventually become
  contiguous before the meeting starts. The order changes only
	through (block) crossings; we bound the number of
  groups that can become supported after each block crossing.

  After a block crossing, at most three pairs of lines that were not
  neighbors before can become neighbors in the permutation: after the
  blocks $C_1, C_2 \subseteq C$ cross, there is one position in
  the permutation where a line of $C_1$ is next to a line of
  $C_2$, and two positions with a line of $C_1$ ($C_2$, respectively)
  and a line of $C \setminus (C_1 \cup C_2)$. Any group
  that was not supported, but is supported
  after the block crossing, must contain one of these pairs.  
  We can describe each such group in the new permutation by specifying
  the new pair and the numbers $d_1$ and $d_2$ of characters of the
  group above and below the new pair in the permutation.
  Since the group size is at most $d$, we have $d_1 + d_2 \le d$.
  The product $d_1 (d-d_1)$ achieves its maximum value for $d_1 = d_2 = d/2$, 
  and so there are at most $d^2/4$ possible groups for each new pair.
  Thus, the total number of newly supported groups after a block crossing is at most 
  $3 d^2/4$, which shows that the optimal number of block crossings is
  at least $4|\hedgespp|/(3d^2)$, completing the proof.
\end{proof}

We now bound the loss of optimality caused by not knowing the initial permutation
used by the optimal solution.
The key idea here is to use a constant-factor approximation for the
problem of deleting the minimum number of hyperedges from $\hyper$ so
that it becomes an interval hypergraph (\textsc{Interval Hypergraph
Edge Deletion}).
We prove the following theorem in
Section~\ref{sec:hyperedge-removal-approx}.

\begin{theorem}
  We can find a $(d+1)$-approximation for \textsc{Interval Hypergraph Edge
  Deletion} on group hypergraphs with $n$ meetings of rank~$d$ in $O(n^2)$ time.
  \label{thm:int-hyper-edge-removal-approx-runtime}
\end{theorem}

Let $\hedges_{\opt}$ be the set of paid edges in the optimal solution,
and $\hedgespp$ the set of paid edges in our algorithm. 
By Theorem~\ref{thm:int-hyper-edge-removal-approx-runtime}, we have
$|\hedgespp| \;\le\; (d+1)|\hedges_{\opt}|$.
Let $\alg$ and $\opt$ be the numbers of block crossings for our algorithm and the 
optimal solution, respectively.
By Lemma~\ref{lem:hyper-story-lower-bound}, we have 
$\opt \: \ge \: 4|\hedges_{\opt}|/(3d^2)$, which gives
$|\hedges_{\opt}| \:\le\: 3d^2/4 \cdot \opt$. 
On the other hand, we have $\alg \:\le\: 2(d-1)|\hedgespp| \:\le\: 2(d-1)(d+1) |\hedges_{\opt}|$.
Combining the two inequalities, we get
$\alg \:\le\: 3(d^2-1)d^2/2 \cdot \opt$, which establishes our main result.

\begin{theorem}
  \label{thm:storyline-gen-norep-approx}
  $d$-SBCM admits a $(3(d^2-1)d^2/2)$-approximation algorithm.
\end{theorem}

\paragraph{Remark.}
We assumed that each group meets only once, but we can extend 
the result if each group can meet $c$ times, for constant $c$.
Our algorithm then yields a $(c \cdot 3(d^2-1)d^2/2)$-factor
approximation; each repetition of a meeting may trigger a
constant number of block crossings not present in the optimal
solution.

\paragraph{Runtime Analysis.}
We have to consider the permutation (of length $k$) of characters
before and after each of the $n$ meetings, as well as after each of
the $O(n)$ block crossings.  This results in $O(kn)$ time for the last
part of the algorithm, but this is dominated by the time ($O(n^2)$)
needed for finding~$\hedgespp$ and for determining the start permutation.

We can improve the running time to $O(kn)$ by a slight modification:
using the
approximation algorithm for \textsc{Interval Hypergraph Edge Deletion} is only
necessary for sparse instances. If $\hyper$ has
sufficiently many edges, any start permutation will yield a good
approximation. Since no meeting
involves more than $d$ characters, no start permutation can support
more than $dk$ meetings. If $n \ge 2dk$, then even the optimal
solution must therefore remove at least half of the edges. Hence,
taking an arbitrary start permutation yields an approximation factor of
at most $2 < d+1$.

We now change the algorithm to use an arbitrary start permutation if $n
\ge 2dk$ and only use the approximation for \textsc{Interval Hypergraph Edge
Deletion} otherwise, i.e., especially only if there are
$O(k)$ edges. Hence, for sparse instances we have $O(n^2) =
O(k^2)$, and for dense instances, the $O(n^2)$ runtime is not
necessary. We get the following improved result.
(The runtime is worst-case optimal since the output complexity is of
the same order.)

\begin{theorem}
  \label{thm:storyline-gen-norep-approx-faster}
  $d$-SBCM admits an $O(kn)$-time $(3(d^2-1)d^2/2)$-approximation
  algorithm.
\end{theorem}

Using some special properties of the 2-character case, we can
improve the approximation factor for 2-SBCM from 18 to 12; see
Appendix~\appref{app:2-SBCM-approx}.

\section{Interval Hypergraph Edge Deletion}
\label{sec:hyperedge-removal-approx}
We now describe the main missing piece from our approximation
algorithm: how to approximate the
minimum number of edges whose deletion reduces a hypergraph to an interval 
hypergraph, i.e., how to solve the following problem.
\begin{problem}[\textsc{Interval Hypergraph Edge Deletion}]
  Given a hypergraph $\hyper = (V,~\hedgese)$ find a smallest set
  $\hedgesepp \subseteq \hedgese$ such that $\hypersp =
  (V,~\hedgese \setminus \hedgesepp)$ is an interval hypergraph.
\end{problem}

Note that a graph contains a Hamiltonian path if and only if one can
remove all but $n-1$ edges so that only vertex-disjoint paths (here, a
single path) remain; hence, our problem is hard even for graphs.
\begin{theorem}
  \label{thm:hypergraph-hardness}
  \textsc{Interval Hypergraph Edge Deletion} is NP-hard.
\end{theorem}

We now present a $(d+1)$-approximation algorithm for rank-$d$
hypergraphs, in which each hyperedge has at most $d$ vertices.
In this section we give all main ideas. Detailed proofs
can be found in Appendix~\appref{sec:hyperedge-removal-approx-appendix};
they are mostly not too hard to obtain, but require the
distinction of many cases.

For our algorithm, we use the following characterization: A hypergraph
is an interval hypergraph if and only if it
contains none of the hypergraphs shown in
Fig.~\ref{fig:forbidden-subhypergraphs} as a
subhypergraph~\cite{trotter1976characterization,MooreJr1977173}.
\begin{figure}[tb]
  \centering
  \includegraphics{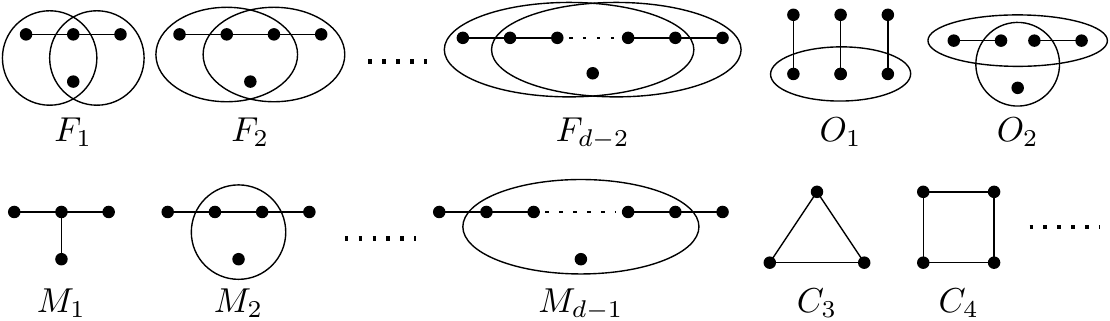}
  \caption{Forbidden subhypergraphs for interval hypergraphs
	(edges represent pairwise hyperedges, circles/ellipses show
	hyperedges of higher cardinality).}
  \label{fig:forbidden-subhypergraphs}
\end{figure}
Due to the bounded rank, the families of $F_k$ and $M_k$ are finite
with $F_{d-2}$ and $M_{d-1}$ as largest members. Cycles are
the only arbitrarily large forbidden subhypergraphs in our
setting.
Let $\forbiddenh = \left\{  O_1, O_2, F_1, \ldots, F_{d-2}, M_1,
\ldots, M_{d-1}, C_3, \ldots, C_{d+1}\right\}$. A hypergraph is
\ffree if it does not contain any hypergraph of $\forbiddenh$ as a
subhypergraph. Note that a cycle in a hypergraph consists of hyperedges $e_1, \ldots, e_k$
so that there are vertices $v_1, \ldots, v_k$ with $v_i \in
e_{i-1} \cap e_{i}$ for $2 \le i \le k$ (and $v_1 \in e_1 \cap
e_{k}$) and no edge $e_i$ contains a vertex of $v_1, \ldots, v_k$
except for $v_i$ and $v_{i+1}$.

Our algorithm consists of two steps. First, we search
for subhypergraphs contained in $\forbiddenh$, and remove all edges involved
in these hypergraphs. In the second step, we break remaining (longer) cycles by
removing some more hyperedges after carefully analyzing the structure
of connected components.
Subhypergraphs in \forbiddenh consist of at most $d+1$
hyperedges. A given optimal solution must remove at
least one of the hyperedges; removing all of them instead yields a factor
of at most $d+1$. The second step will not negatively affect this
approximation factor.

Intuitively, allowing long cycles, but forbidding
subhypergraphs of $\forbiddenh$, results in a generalization
of interval hypergraphs where the vertices may be placed on a cycle
instead of a vertical line. This is not exactly true, but we will see
that the connected components after the first step have a structure similar to
this, which will help us find a set of edges whose
removal destroys all remaining long cycles.

Lemma~\appref{lemma:hyper-cycle-shape}
(Appendix~\appref{sec:hyperedge-removal-approx-appendix}) shows that any
vertex is contained in at most three
hyperedges of a cycle, where the case of three hyperedges with a
common vertex occurs only if a hyperedge is contained in the union of
its two neighbors in the cycle.
Assume that $e_1, e_2$, and $e_3$ are consecutive
edges of a cycle $C$. If all three edges are present in an
interval representation, we know
that we will first encounter vertices that are only
contained in $e_1$, then vertices that are in $(e_1 \cap e_2) \setminus
e_3$, then vertices in $e_1 \cap e_2 \cap e_3$, followed by vertices
of $(e_2 \cap e_3) \setminus e_1$, and vertices of $e_3 \setminus (e_1
\cup e_2)$. Some of the sets (except for pairwise intersections) may
be empty. We do not know the order of vertices
within one set, but we know the relative order of any pair
of vertices of different sets.
By generalizing this to the
whole cycle, we get a cyclic order---describing the local order in a
possible interval representation---of sets defined by containment in
1, 2, or 3 hyperedges. We call these sets \emph{cycle-sets} and their
cyclic order the \emph{cycle-order} of $C$.

We can analyze how an edge~$e \notin C$ relates to the order of
cycle-sets; $e$ can contain a cycle-set
completely, can be disjoint from it, or can contain only part of
its vertices. We call a consecutive sequence of cycle-sets contained
in edge~$e$---potentially starting and ending with cycle-sets
partially contained in $e$---an \emph{interval} of $e$ on $C$.
The following lemma shows that every edge forms only a single interval
on a given cycle.

\begin{lemma}
	\label{lemma:cycle-edge-interval-relation}
  If a hyperedge $e \in \hedgese$ intersects two cycle-sets of a
	cycle~$C$, then $e$ fully contains all cycle-sets lying in between
	in one of the two directions along~$C$.
\end{lemma}

We now know that by opening the cycle at a single position within
a cycle-set not contained in $e$,
$C+e$ forms an interval hypergraph. Edge $e$ adds further
information: If only part
of the vertices of a cycle-set are contained in $e$ and also vertices
of the next cycle-set in one direction, we know that the vertices of
$e$ in the first cycle-set should be next to the second cycle-set.
We use this to refine the cycle-sets to a cyclic
order of \emph{cells}, the \emph{cell order} (a cell is a set of
vertices that should be contiguous in the cyclic order). Initially,
the cells are the
cycle-sets. In each
step we refine the cell-order by inserting an edge containing vertices
of more than one cell, possibly splitting two cells into two subcells
each. The following lemma shows that during this process of
refinements, as an invariant each remaining edge forms a single
interval on the cell order.

\begin{lemma}
	\label{lemma:cell-edge-interval-relation}
  If a hyperedge $e \in \hedgese$ intersects two cells, then $e$
  fully contains all cells lying in between in one of
  the two directions along the cyclic order.
\end{lemma}

After refining cells as long as possible,
each edge of the connected component that we did not insert lies
completely within a single cell. Several edges can lie within the same
cell, forming a hypergraph that imposes restrictions on the order of
vertices within the cell.  However, the cell contains fewer than $d$
vertices. Hence, this small hypergraph cannot contain any cycles,
since we removed all short cycles, and must be an interval hypergraph.

With this cell-structure, it is not too hard to show that the
following strategy to make the connected component an interval
hypergraph is optimal (see Lemmas~\appref{lemma:cell-neighboring-cycle},
\appref{lemma:cell-neighboring-break}
and~\appref{lemma:cell-neighboring-cycle-destruction} in
Appendix~\appref{sec:hyperedge-removal-approx-appendix}): For each pair
of adjacent cells we
determine the number of edges containing both cells, select the pair
minimizing that number, and remove all edges containing both.
The cell order then yields an order of the connected
component's vertices that supports all remaining edges.
Since this last step of the algorithm is done optimally, we do not
further change the approximation ratio, which, overall, is $d+1$,
because we never remove more than $d+1$ edges for at least one edge that
the optimal solution removes.

\paragraph{Runtime.} 
Our algorithm can be implemented to run in $O(m^2)$ time for
$m$ hyperedges. We give the
main ideas here and present details in
Appendix~\appref{sec:hyperedge-deletion-runtime}.
When searching for forbidden subhypergraphs, we
first remove all cycles of length
$k \le d$ using a modified breadth-first search in $O(m^2)$ time.
The remaining types of forbidden subhypergraphs each contain an edge
that contains all but one ($O_2$ and $F_k$), two ($M_k$), or three
($O_1$) vertices of the 
subhypergraph. We always start searching from such an edge and use
that all short cycles have already been removed.
In the second phase, we determine the connected components and initialize the cell
order for each of them, in $O(n+m)$ time. Stepwise refinement requires $O(m^2)$ time.
Counting hyperedges between adjacent cells, determining
optimal splitting points, and finding the final order can all
be done in linear time.

\begin{repeattheorem}{thm:int-hyper-edge-removal-approx-runtime}
  We can find a $(d+1)$-approximation for \textsc{Interval Hypergraph Edge
  Deletion} on hypergraphs with $m$ hyperedges of rank~$d$ in $O(m^2)$ time.
\end{repeattheorem}

\bibliographystyle{splncs03}
\bibliography{slbc,abbrv,lit}

\newpage
\appendix
\chapter*{\appendixname}

\section{Preliminaries: Proofs}
\label{app:preliminaries}

\begin{repeatobservation}{obs:badupper}
  \contentObsBadupper
\end{repeatobservation}
\begin{proof}
  Let $\pi'$ be an arbitrary permutation and $m = \{c, c'\} \in M$ the next
  meeting.
  Let $i$ and $j$ be the positions of the characters in the permutation, that
  is, $\pi'_i = c$ and $\pi'_j =c'$.
  Without loss of generality, assume $i < j$.
  If $\pi'$ does not support $m$, we can realize it using the block
	crossings $(i,i,j-1)$, that is, moving the line of $c$ directly
	above that of $c'$.
\end{proof}

\begin{proposition}
  \label{prop:counterexample}
  There is an instance $\CS$ of 2-SBCM and a start permutation~$\pi^0$
  such that there is no optimal solution $(\pi^0, B)$ of $\CS$ that
  starts with $\pi^0$ and uses at most one block crossing before the
  first and between each pair of consecutive meetings.
\end{proposition}
\begin{proof}[by contradiction]
  Consider the instance $\CS = (C, M)$ with
  \begin{align*}
    C &= \{1, 2, 3, 4, 5, 6, 7, 8\} \text{ and}  \\
    M &=[\{6,3\}, \{7,2\}, \{1,5\}, \{5,6\}, \{6,3\}, \{3,4\}, \{4,8\}, \{8,7\}].
  \end{align*}
  Let $\pi^0=\langle 1,2,3,4,5,6,7, 8 \rangle$ be the
  start permutation.  There is a solution that performs only two block
  crossings, namely $(\pi^0, B)$ with
  $B = [(2,4,7), (4,5,8)]$, see Fig.~\ref{fig:counterexample}.
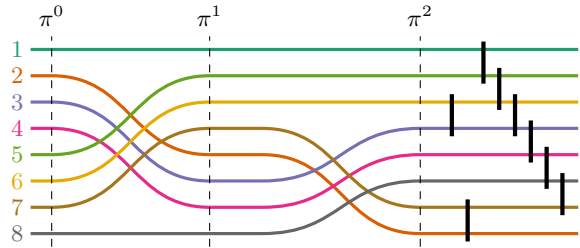
\begin{figure}[tb]
    \centering
    \begin{tikzpicture}[scale=.7]
      \storylineset{
        norightlabel,
        leftwidth=0.1,
        rightwidth=0.75,
        distx=4,
        crossinglength=0.75
      }

      \drawstoryline[drawingstyle=char1,label=$1$]{8,8,8}
      \drawstoryline[drawingstyle=char2,label=$2$]{7,4,1}
      \drawstoryline[drawingstyle=char3,label=$3$]{6,3,5}
      \drawstoryline[drawingstyle=char4,label=$4$]{5,2,4}
      \drawstoryline[drawingstyle=char5,label=$5$]{4,7,7}
      \drawstoryline[drawingstyle=char6,label=$6$]{3,6,6}
      \drawstoryline[drawingstyle=char7,label=$7$]{2,5,2}
      \drawstoryline[drawingstyle=char8,label=$8$]{1,1,3}

      \drawstorylinemeeting{1.9}{5}{6}
      \drawstorylinemeeting{1.975}{1}{2}
      \drawstorylinemeeting{2.05}{7}{8}
      \drawstorylinemeeting{2.125}{6}{7}
      \drawstorylinemeeting{2.2}{5}{6}
      \drawstorylinemeeting{2.275}{4}{5}
      \drawstorylinemeeting{2.35}{3}{4}
      \drawstorylinemeeting{2.425}{2}{3}

      \drawstorylinepermutation{0}{1}{8}{$\pi^0$}
      \drawstorylinepermutation{0.75}{1}{8}{$\pi^1$}
      \drawstorylinepermutation{1.75}{1}{8}{$\pi^2$}
    \end{tikzpicture}
    \caption{Optimal solution for $\CS$ from the proof of
      Proposition~\ref{prop:counterexample}.}
    \label{fig:counterexample}
\end{figure}
  Let $\pi^1$ be the permutation after the first block crossing of $B$
  on $\pi^0$, and $\pi^2$ the permutation after both block
  crossings.
  The permutation $\pi^2$ supports all meetings in $M$.
  The first meeting $\{6,3\}$ in $M$ does not fit $\pi^0$ or $\pi^1$, that is,
  both block crossings occur before the first meeting.

  Now assume there is another solution $(\pi^0, B')$ with $|B'| \leq 2$ that
  has at most one block crossing before each meeting.
  Starting from $\pi^0$ there are exactly nine feasible block
  crossings that allow the first meeting.
  They yield the following permutations:
  \begin{multicols}{2}
    \begin{itemize}
      \item $\langle 1, 2, 4, 5, 6, 3, 7, 8 \rangle$
      \item $\langle 1, 2, 5, 6, 3, 4, 7, 8 \rangle$
      \item $\langle 1, 2, 6, 3, 4, 5, 7, 8 \rangle$
      \item $\langle 4, 5, 1, 2, 3, 6, 7, 8 \rangle$
      \item $\langle 1, 4, 5, 2, 3, 6, 7, 8 \rangle$
      \item $\langle 1, 2, 4, 5, 3, 6, 7, 8 \rangle$
      \item $\langle 1, 2, 3, 6, 4, 5, 7, 8 \rangle$
      \item $\langle 1, 2, 3, 6, 7, 4, 5, 8 \rangle$
      \item $\langle 1, 2, 3, 6, 7, 8, 4, 5 \rangle$
    \end{itemize}
  \end{multicols}
  None of these permutations supports the second meeting $\{7, 2\}$.
  So we need the second block crossing before this meeting.
  This second block crossing needs to prepare all of the remaining meetings,
  because otherwise $|B'| > 2$.
  These meetings can only be supported by the permutation $\sigma = (1, 5, 6,
  3, 4, 8, 7, 2)$ or its reverse permutation $\sigma^R$.
  It remains to show that none of the permutations yielded by the feasible
  first block crossing can be transformed to $\sigma$ or $\sigma^R$ by one
  additional block crossing.
  All permutations containing $\langle 3,6 \rangle$ as a subsequence are
  infeasible because there is only one block crossing that swaps two neighboring
  characters and it does not produce $\sigma$.
  For permutations starting with $\langle 1, 2 \rangle$ there is only one
  possible block crossing to bring $2$ to the end of the permutation while $1$ stays
  at the first position,
  which also does not yield $\sigma$.
  Similarly, we can show that there is also no block crossing after any of the
  feasible block crossing for the first step that leads to $\sigma^R$.
\end{proof}

\section{NP-Hardness without Repetitions}
\label{sec:hardness-no-rep}
With arbitrarily large meetings, we can slightly modify our hardness
proof, and show that minimizing the number of block crossings is also
hard without repeating the same meeting many times. The idea to
change our reduced instance, is to replace the repeated sequence of
2-character meetings so that in each repetition the group size is
increased by one for all meetings; see
Fig.~\ref{fig:hardness-general-meetings}.
\begin{figure}[b]
  \centering
  \includegraphics[scale=1.4]{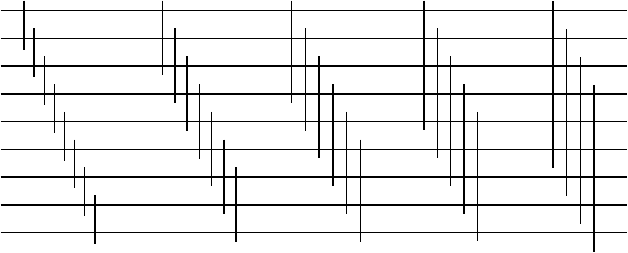}
  \caption{Simulating repeated 2-character meetings using groups of
  increasing size}
  \label{fig:hardness-general-meetings}
\end{figure}

Due to the overlapping structure of the groups in a single sequence,
they can only be all supported at the same time if also the
2-character
meetings that they replaced are supported. The only thing that we have
to be careful about is that when the groups get larger than
$k/2$ there is a growing set of characters in the middle that are
contained in exactly the same groups, and their relative
order does not matter. We will avoid that this happens.

Since we have $k+1$ sequences of repeated meetings at the beginning as
well as at the end of the timeline, and we keep increasing the group
sizes, we have groups of $2k+3$ characters in the end. We
replace $c_1, \ldots, c_{2k}$ by a new sequence $c_1, \ldots,
u_{5k}$ of characters without changing anything else on the structure.
Then, we can increase the group size up to $2k+3$ while in the end
still less than half of all characters are involved in each group.
Since the growing meetings completely simulate the desired 2-character
meetings, the rest of the reduction and its proof stay the same, and
we get the following result.
\begin{theorem}
  SBCM is NP-hard even if meetings are not repeated.
  \label{thm:general-norepeat-np-hard}
\end{theorem}

\section{Exact Algorithms: Proofs}
\label{app:exact}

\begin{repeatlemma}{lem:datastruct}
  \contentLemDatastruct
\end{repeatlemma}
\begin{proof}
Represent the permutation as a doubly-linked list. %
Then it takes constant time to check whether a 2-meeting fits: check the previous/next pointers.
Since a block crossing changes at most 6 adjacencies, a \method{BlockMove} can update the linked list in constant time.

Now we look at a meeting of cardinality $m$.
Interpret the linked list as a path and consider the subgraph induced by the nodes in the meeting.
If the meeting fits the permutation, this subgraph is connected and, being a path, has $m-1$ edges; if the meeting does not fit, this subgraph has more components and therefore fewer edges.
The \method{Check} operation on a meeting of size $m$ can be performed in $O(m)$ time by counting at every node in the meeting whether zero, one or two of its neighbors are also in the meeting.
For the amortized runtime over a sequence of operations, remember this count: \method{BlockMove} can update it in constant time, since again at most 6 adjacencies change.

In terms of space, there is only the doubly linked list and the count.
\end{proof}

\section{SBCM with Meetings of Two Characters: Proofs}
\label{app:2sbcm}

For the following lemma, we assume that no two subsequent meetings in the input are the same.
We call an instance \emph{normal} if this is the case.
An instance can be normalized by simply dropping the repeated meetings.
This does not affect the optimum number of block crossings or the behavior of the greedy algorithm, but note that it does lower $n$.

\begin{lemma}\label{lem:2sbcmUpper}
A normal instance of 2-SBCM with $k=3$ can be solved using at most
$\lceil n/2 \rceil-1$ block crossings.
\end{lemma}

\begin{proof}
Note that there are only three possible meetings, namely $\{1,2\}$, $\{1,3\}$, and $\{2,3\}$.
Any permutation supports precisely two of these and not the third, and is equivalent in this sense to its reverse.
For example, the permutation $\langle1,2,3\rangle$ and its reverse support the meetings $\{1,2\}$ and $\{2,3\}$, but not $\{1,3\}$.
Let $\pi$ and $\pi'$ be distinct permutations.
Case distinction shows that it is always possible in a single block
crossing to get from $\pi$ to either $\pi'$ or its reverse.

For the analysis, we partition the sequence of meetings into \emph{epochs} as follows.
We start from the first meeting and keep going until the third distinct meeting occurs: these meetings form the first epoch.
That is, an epoch alternates between two different meetings.
Repeating this process partitions the entire sequence of meetings into epochs, possibly with a single remaining meeting as final epoch.
A solution can choose the start permutation $\pi^0$ that supports the first epoch.
After that it can always get to a permutation that supports the entire
next epoch in one block crossing.
In the worst case all epochs have length~2, and we need $\lceil n/2
\rceil-1$ block crossings.
\end{proof}

\begin{repeattheorem}{thm:k3-greedy-optimal}
  \contentThmGreedyOptimal
\end{repeattheorem}
\begin{proof}
We look at the epochs from Lemma~\ref{lem:2sbcmUpper} again.
The greedy algorithm produces one block crossing fewer than the number of epochs.

Consider any epoch except the last one and include the meeting after it.
By construction, this is the third distinct meeting and therefore these meetings together cannot fit a single permutation.
Then in any solution to the problem, a block crossing must occur after at least one of the meetings in the epoch.
This holds for all epochs except the last one and since they are
disjoint, the number of epochs reduced by one is a lower bound for the
optimum number of block crossings.
The result of the greedy algorithm realizes this bound.
\end{proof}

\section{Improved Approximation for 2-SBCM}
\label{app:2-SBCM-approx}
By using specific
structures for 2-character meetings we can improve approximation
factor and runtime (the general algorithm yields an
18-approximation).

Note that for 2-character meetings the group hypergraph is a graph,
and an interval hypergraph here is a collection of
vertex-disjoint paths. Our algorithm for \textsc{Interval Hypergraph
Edge Deletion} for $d=2$ yields a 3-approximation.
We develop a better approximation using the following
observation. Consider a character $c$ in the collection of paths
supported in the beginning of some solution. If $c$ has two neighbors
$c_1$ and $c_2$ in its path, but $c$'s first meeting is with a
character $c_3 \notin \{c_1, c_2\}$, then at the beginning of that
meeting $c$ can only be neighbor to one of the two, say, to $c_1$, even
in an optimal solution; the meeting with $c_2$ then must later be
reconstructed by block crossings. Hence, the
effective set of meetings supported in the beginning is in fact a collection
of paths with the additional restriction that each character is adjacent
to at most one character except for the one he meets first.
Without changing the rest of the analysis, we can approximate this new
problem for finding the start permutation.

We first consider, for each vertex $c$, all edges incident to $c$ except
for the one describing $c$'s first meeting. If there are $\ell \ge 2$ such
edges, we know that even the optimal solution can support at most one
of them and, hence, has to remove $\ell-1$ of them. We remove all
$\ell$ of them, which yields an approximation factor of $\ell/(\ell-1) \le 2$.
Eventually, all vertices have degree~2 or less and the connected
components are paths and cycles. For each cycle, we remove one arbitrary
edge, so that we end up with a collection of paths. This second step
does not change the approximation factor since the optimal solution
has to remove at least one edge per cycle as well. This
algorithm easily runs in linear time,
which speeds up the runtime of the complete algorithm to $O(kn)$.

\begin{theorem}
  \label{thm:storyline-pairwise-norep-approx-improved}
  We can find a $12$-approximate solution
  for $2$-SBCM
  without repetitions in $O(kn)$ time.
\end{theorem}

\section{Interval Hypergraph Edge Deletion}
\label{sec:hyperedge-removal-approx-appendix}
\begin{lemma}
  Let $\hyper = (V, \hedgese)$ be an \ffree hypergraph. Let $C$ be a
  cycle appearing as a subhypergraph in $\hyper$. Then two edges of
  $C$ have a common vertex if and only if they are consecutive in $C$
  or they share a common neighbor in $C$.
  \label{lemma:hyper-cycle-shape}
\end{lemma}
\begin{proof}
  No edge of the $C$ can fully contain another edge of $C$.
	Let $e_1, e_2, e_3 \in C$  be three
  edges of $C$, and assume that $e_1 \cap e_2 \cap e_3 \supseteq
  \left\{ v \right\} \neq \emptyset$. If there are vertices $v_1 \in
  e_1 \setminus (e_2 \cup e_3)$, $v_2 \in e_2 \setminus (e_1 \cup
  e_3)$, and $v_3 \in e_3 \setminus (e_1 \cup e_2)$, the three
  hyperedges form a subhypergraph of type $M_1$ (with $v_1, v_2, v_3$,
  and $v$ serving as vertices); see
  Fig.~\ref{fig:hypergraph-cycle-edge-structure-claw}.
  \begin{figure}[tb]
    \begin{subfigure}[t]{.25\textwidth}
      \centering
      \includegraphics{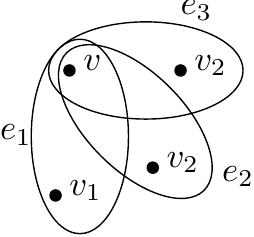}
      \caption{3 hyperedges forming $M_1$.}
      \label{fig:hypergraph-cycle-edge-structure-claw}
    \end{subfigure}
    \hfill
    \begin{subfigure}[t]{.31\textwidth}
      \centering
      \includegraphics[page=2]{img/hypergraph-cycle-edge-structure}
      \caption{3 edges consecutive on $C$.}
      \label{fig:hypergraph-cycle-edge-structure-consecutive}
    \end{subfigure}
    \hfill
    \begin{subfigure}[t]{.36\textwidth}
      \centering
      \includegraphics[page=3]{img/hypergraph-cycle-edge-structure}
      \caption{$O_1$ as a subhypergraph.}
      \label{fig:hypergraph-cycle-edge-structure-o3}
    \end{subfigure}
    \caption{Illustrations of the proof of Lemma~\ref{lemma:hyper-cycle-shape}.}
    \label{fig:hypergraph-cycle-edge-structure}
  \end{figure}

  On the other hand, if one of the
  three, say, $v_2$ does not exist, we have $e_2 \subseteq e_1 \cup
  e_3$ and one easily checks that this can only be the case if the
  three edges are consecutive on the cycle, since every vertex of
  $e_2$ must also be a vertex of $e_1$ or $e_3$; see
  Fig.~\ref{fig:hypergraph-cycle-edge-structure-consecutive}.

  Now, assume that there are two edges $e, e' \in C$ with $e \cap e'
  \supseteq \left\{ v \right\} \neq \emptyset$ that are neither
  consecutive nor have a common neighboring hyperedge in $C$. As we
  have seen, $v$ can be contained in none of the neighbors of
	$e$ and $e'$ in $C$. Let $e_1$ and $e_2$ be the neighbors of
	$e$ in $C$. If either of the two intersects with $e'$, we find
	$C_3$ as a subhypergraph, a contradiction. Hence, there are elements
	$v_1 \in e_1$, $v_2 \in e_2$, and $v' \in e'$ so that each of the
	vertices is contained in no other of the four involved hyperedges.
	With these vertices, we have found $O_1$ as a subhypergraph; see
  Fig.~\ref{fig:hypergraph-cycle-edge-structure-o3}.
\end{proof}

With this lemma, we know about the structure of the vertices contained
in  hyperedges of a cycle: A vertex can be contained in at most three
hyperedges of the cycle, where the case of three hyperedges with a
common vertex occurs only if a hyperedge is contained in the union of
its two neighbors in the cycle.

Assume that $e_1, e_2$, and $e_3$ are three consecutive
edges of a cycle $C$. If all three edges are present in an
interval representation for part of the edges of \hyper, we know
that in the order we will first encounter vertices that are only
contained in $e_1$, then vertices that are in $(e_1 \cap e_2) \setminus
e_3$, then vertices that are in $e_1 \cap e_2 \cap e_3$, followed by vertices of $(e_2 \cap e_3) \setminus e_1$, and vertices of $e_3 \setminus (e_1
\cup e_2)$. Some of these sets (except for the pairwise intersections) may
be empty. We do not know anything about the relative order of vertices
within one of these sets, but we know the relative order of any pair
of vertices of different sets; see Fig.~\ref{fig:cycle-sets-order}.
\begin{figure}[tb]
  \centering
  \includegraphics{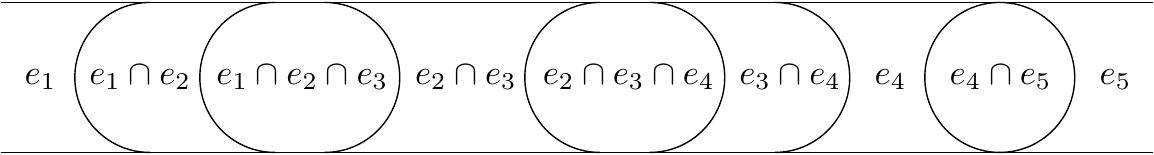}
  \caption{Cycle-sets and their relative order.}
  \label{fig:cycle-sets-order}
\end{figure}
By generalizing this to the
whole cycle, we get a cyclic order---describing the local order in a
possible interval representation---of sets defined by containment in
1, 2, or 3 hyperedges. We call these sets \emph{cycle-sets}, and their
cyclic order the \emph{cycle-order} of $C$.

\begin{lemma}
  Let $\hyper = (V, \hedgese)$ be an \ffree hypergraph and let $C$ be a
  cycle appearing as a subhypergraph in $\hyper$. There is no
  hyperedge $e \in \hedgese$ that contains both vertices of edges of
  $C$ and at least one vertex $v \notin \bigcup_{e' \in C} e'$.
	\label{lemma:hypergraph-two-cycles}
\end{lemma}
\begin{proof}
	Assume to the contrary that such a hyperedge $e$ exists. If
	$e$ contains at least one vertex that lies in the intersection of
	two edges of $C$, then we find a subhypergraph $M_k$
	(with a $k \le d-1$) as follows. Assume $v' \in e \cap e_1 \cap e_2$
	with edges $e_1, e_2$ consecutive on $C$. From $v'$ on we
	follow $C$ in both directions as long as as we find vertices in the
	intersection of consecutive cycle edges that also belong to
	$v$. This process must stop eventually, since $e$ can contain at
	most $d-1$ vertices of cycle edges, while $C$ has length at
	least $d+2$. Together with two more vertices of the next intersections of
	cycle edges (that are not in $e$), we have found a path that, with
	$e$ and $v$, forms a subhypergraph of the type $M_k$; see
  Fig.~\ref{fig:hypergraph-edge+cycle}
  \begin{figure}[tb]
    \centering
    \includegraphics{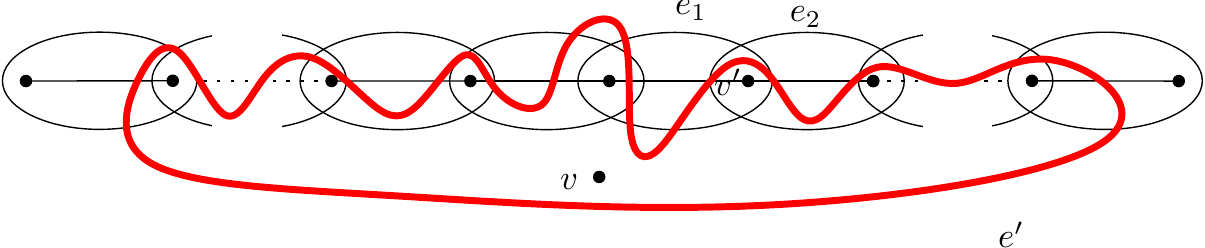}
    \caption{$M_k$ as a subhyperedge if $e'$ contains vertices
    involved in the cycle.}
    \label{fig:hypergraph-edge+cycle}
  \end{figure}

	Now, we know that $e$ cannot contain a vertex that lies in two
	cycle edges, but there can still be an edge $e'$ of $C$ with
	a vertex $v' \in e \cap e'$. However, by using $e$, $e'$, and the
	two neighbors of $e'$ in $C$ we immediately find $O_1$ as a
  subhypergraph (just as in
  Fig.~\ref{fig:hypergraph-cycle-edge-structure-o3}).
\end{proof}

As a consequence of the previous lemma, the hyperedges of two
different cycles either cover the exactly same set of vertices, or
their sets of vertices are disjoint. This also means that each
connected component is either acyclic, or forms a ground for a set of
cycles. We now try to analyze the structure of cycles on such a
connected component in order to break all remaining cycles optimally.

If two cycles share their vertex sets, we can analyze how an edge of
the one cycle relates to the structure---the cycle sets and their
order---of the other cycle. Recall that we know about the relative
order of the cycle-sets, but not of the internal order of vertices
within the same cycle-set. Another edge can contain a cycle-set
completely, can be disjoint from it, or can contain only part of
its vertices. We call a consecutive sequence of cycle-sets contained
in edge~$e$---potentially starting and ending with cycle-sets
partially contained in $e$---an \emph{interval} of $e$ on $C$.
The following lemma shows that every edge forms only a single interval
on a given cycle.

\begin{repeatlemma}{lemma:cycle-edge-interval-relation}
  Let $\hyper = (V, \hedgese)$ be an \ffree hypergraph.
  Let $C$ be a cycle appearing as
  a subhypergraph in $\hyper$ and let $e \in
	\hedgese$ be a hyperedge on the same vertex set $\bigcup_{e' \in C} e'$.
  If $e$ intersects two cycle-sets, then $e$ must fully contain the
  vertices of all cycle-sets lying in between in one of the two
  directions along the cycle.
\end{repeatlemma}
\begin{proof}
	Assume that the claim is not true, that is, $e$ consists of a
  collection of at least two intervals of (partially) contained
  cycle-sets, where any two such intervals are separated by a vertex
  not in $e$ lying in a cycle-set. We distinguish cases similar to the
  proof of the previous lemma. First, assume that one such interval
  contains a vertex of the cycle. We follow the cycle in both
  directions from that vertex, as long as we find a vertex of
  $e$ in the intersection of the current edge with the next one along
  the cycle. Since $e$ has at most $d$ vertices but $C$ has length at
  least $d+2$, this process will eventually stop, thus forming a path
  of length at least two, whose first and last vertices are not in
  $e$, but all internal vertices are. Now, assume that there is
  another vertex $v \in e$ that is contained in none of the edges of
  the path. Then, we have found $M_k$ as a subhypergraph. 

  On the other hand, if there is no such vertex $v$, we still know that
  there must be more than one interval formed by $e$. Hence, there
  especially must be a vertex $v'$ in a cycle-set separating two
  consecutive internal vertices $v_1, v_2$ of the path that is not contained in
  $e$. Let $e'$ be the edge of $C$ connecting $v_1$ and $v_2$; see
  Fig.~\ref{fig:cycle-structure-interval-case1}.
  \begin{figure}[tb]
    \begin{subfigure}[t]{.32\textwidth}
      \centering
      \includegraphics{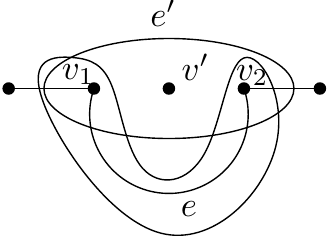}
      \caption{$M_2$ as a subhypergraph.}
      \label{fig:cycle-structure-interval-case1}
    \end{subfigure}
    \hfill
    \begin{subfigure}[t]{.32\textwidth}
      \centering
      \includegraphics[page=2]{img/cycle-structure-interval-lemma}
      \caption{$M_1$ as a subhypergraph.}
      \label{fig:cycle-structure-interval-case2}
    \end{subfigure}
    \hfill
    \begin{subfigure}[t]{.34\textwidth}
      \centering
      \includegraphics[page=3]{img/cycle-structure-interval-lemma}
      \caption{$O_2$ as a subhypergraph.}
      \label{fig:cycle-structure-interval-case3}
    \end{subfigure}
    \caption{Vertex $v' \notin e$ in a gap between two intervals of
    $e$.}
    \label{fig:cycle-structure-interval-cases}
  \end{figure}
  If none of the neighbors of $v_1$ and $v_2$ along the cycle lies in
  $e$, we have found a $M_2$-subhypergraph as in
  Fig.~\ref{fig:cycle-structure-interval-case1}.
  If the neighbor of only one of them, say, $v_1$ is in $e$ but the
  neighbor of $v_2$ isn't, then by disregarding $v_1$ we find an
  $M_1$-subhypergraph centered on $v_2$; see
  Fig.~\ref{fig:cycle-structure-interval-case2}. On the other
  hand, if both neighbors lie in $e$, then we have $O_2$ as a
  subhypergraph; see Fig.~\ref{fig:cycle-structure-interval-case3}.

  If $e$ contains no element in the intersection of any two
  consecutive cycle edges, then we take vertices $v$ and $v'$ from
  two different intervals; $v \in e_1$ and $v' \notin e_1$ for a cycle
  edge $e_1$. Let $e_0$ and $e_2$ be the neighbors of $e_1$ in
  $C$. We have, $v' \notin e_0 \cup e_2$ since otherwise there would be a
  triangle. Then, $e_0$, $e_1$, $e_2$, and $e$ (via $v'$) form $O_1$
  as a subhypergraph.
\end{proof}

Since $e$ forms only a single interval of cycle-sets, we know that by
opening the cycle at a single position within a cycle-set not
contained in $e$,
$C+e$ forms an interval hypergraph. $e$ adds further
information on the relative order within some cycle-sets. If only part
of the vertices of a cycle-set are contained in $e$ and also vertices
of the next cycle-set in one direction, we know that the vertices of
$e$ in the first cycle-set should be next to the second cycle-set.

We use this to refine the cycle-sets to a cell structure with a cyclic
order of cells, the \emph{cell order}.
A cell is just a set of vertices that must be contiguous in the cyclic
order prescribed by hyperedges.
Initially, the cells are the cycle-sets. Then, in each
step we refine the cell-order by inserting an edge containing vertices
of more than one cell, possibly splitting two cells into two subcells
each. If after refining the cell order, it is still true that each
remaining edge forms a single interval, then this results in a final
refined cell order, where each remaining edge of the connected
component must be fully contained in one of the cells. The following
lemma shows that the interval property is indeed preserved during the
process of refinements.

\begin{repeatlemma}{lemma:cell-edge-interval-relation}
  Let $\hyper = (V, \hedgese)$ be an \ffree hypergraph. Let $C$ be a
  cycle appearing as a subhypergraph in $\hyper$. If we initialize the
  cell order with
  the cycle-sets of $C$ and keep refining the structure by considering
  edges that contain vertices of at least two different cells, then
  the following \emph{interval property} holds for any hyperedge $e
  \in \hedgese$ on the vertex set $\bigcup_{e' \in C} e'$:

	If $e$ intersects two cells, then $e$ must fully contain the vertices of all
  cells lying in between in one of the two directions along the
  cyclic order.
\end{repeatlemma}
\begin{proof}
  We show the property by induction over the insertions. Due to
  Lemma~\ref{lemma:cycle-edge-interval-relation} it holds
  in the beginning. Now, assume that the interval property holds for
  the cell order after inserting a set of edges. We show that after
  refining the cells by considering another edge $e'$, the property
  still holds.

  Assume that for the refined cells the interval property does not hold for an
  edge $e$. Since the property did hold for the cells of the previous
  step, the only problem can occur in a cell $c$ of the previous
  step that is only partially contained by both $e$ and $e'$. Without
  loss of generality, we can assume that $e'$ also contains elements of
  the cell right of $c$; let $c_1$ and $c_2$ in this order be the
  (nonempty) cells resulting from splitting $c$, i.e., $c_1 = c
  \setminus e'$ and $c_2 = c \cap e'$. There are two basic cases in
  which the interval property could be violated for $e$.

  First, if $e$ contains also elements of the cell right of $c$, then
  we have a violation only if there are vertices $v_1 \in c_1 \cap e$
  and $v_2 \in c_2 \setminus e$. We distinguish cases based on the
  right boundary of cell $c$, which---from left to right---can either
  be closing or opening one (or more) hyperedge $\tilde{e}$.
  
  First assume that it is
  closing $\tilde{e}$; $\tilde{e}$ fully contains $c$ and at
  least also the cell left of $c$. If there is a common vertex of
  $e$ and $e'$ in the cell right of $c$, then we find $C_3$ as a
  subhypergraph with $e$, $e'$, and $\tilde{e}$; see
  Fig.~\ref{fig:int-ind-case1-a}. Otherwise, there are vertices in the next
  cell that are unique for $e$ and $e'$, respectively.
  Since we never inserted an edge completely contained in cells, this
  must also have held for $\tilde{e}$. Therefore, there must be an
  edge (apart from $e$ and $e'$) containing some (but not all) cells from
  $\tilde{e}$ and cells either left or right of $\tilde{e}$. If such
  an edge $\bar{e}$ contains cells to the right, then it must
  especially contain cell $c$ and the cell
  right of it. Together with a vertex in $\tilde{e}$ not contained in
  $\bar{e}$, we have found $O_2$ as a subhypergraph;
  see Fig.~\ref{fig:int-ind-case1-b}. On the other hand, if $\bar{e}$ contains cells
  of $\tilde{e}$ and cells left of it, then we find $O_1$ as a
  subhypergraph by adding a vertex in $\tilde{e} \cap \bar{e}$ (not in
  $c$) and a vertex in $\bar{e} \setminus
  \tilde{e}$; see Fig.~\ref{fig:int-ind-case1-c}.
  \begin{figure}[tb]
    \begin{subfigure}[t]{.27\textwidth}
      \centering
      \includegraphics[scale=1.2,page=1]{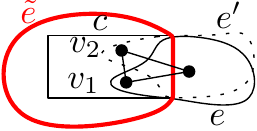}
      \caption{$C_3$.}
      \label{fig:int-ind-case1-a}
    \end{subfigure}
    \hfill
    \begin{subfigure}[t]{.31\textwidth}
      \centering
      \includegraphics[scale=1.2,page=2]{img/interval-lemma-induction-1-sided-cases}
      \caption{$O_2$.}
      \label{fig:int-ind-case1-b}
    \end{subfigure}
    \hfill
    \begin{subfigure}[t]{.33\textwidth}
      \centering
      \includegraphics[scale=1.2,page=3]{img/interval-lemma-induction-1-sided-cases}
      \caption{$O_1$.}
      \label{fig:int-ind-case1-c}
    \end{subfigure}
    \caption{$e$ and $e'$ overlapping in cell $c$ from the same
    direction.}
    \label{fig:int-ind-cases1-1}
  \end{figure}

  Now, assume that $\tilde{e}$ is opening on the right boundary of
  $c$. If the cell right of $c$ contains no common element of
  $e$ and $e'$, the situation is symmetric to the one we had before by
  exchanging the role of $c$ with the cell right of $c$; see
  Fig.~\ref{fig:int-ind-case1-d}. Otherwise,
  there is an element of $e \cap e'$ in the next cell. If
  $\tilde{e}$ contains an element not in $e \cup e'$, then we have
  found $M_1$ as a subhypergraph; see Fig.~\ref{fig:int-ind-case1-e}. We know that
  there must be at least one previously inserted hyperedge
  $\bar{e}$ overlapping with $\tilde{e}$. Assume that $\bar{e}$ is
  overlapping from the left. If there is a vertex of $e \cap e'$ in $\tilde{e}
  \setminus \bar{e}$, we have found a $C_3$-subhypergraph; see
  Fig.~\ref{fig:int-ind-case1-f}.
  \begin{figure}[tb]
    \begin{subfigure}[t]{.32\textwidth}
      \centering
      \includegraphics[scale=1.2,page=4]{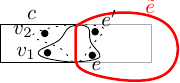}
      \caption{Symmetric situation.}
      \label{fig:int-ind-case1-d}
    \end{subfigure}
    \begin{subfigure}[t]{.32\textwidth}
      \centering
      \includegraphics[scale=1.2,page=5]{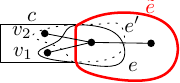}
      \caption{$M_1$.}
      \label{fig:int-ind-case1-e}
    \end{subfigure}
    \begin{subfigure}[t]{.32\textwidth}
      \centering
      \includegraphics[scale=1.2,page=6]{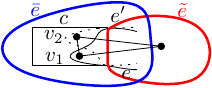}
      \caption{$C_3$.}
      \label{fig:int-ind-case1-f}
    \end{subfigure}
    \caption{$e$ and $e'$ overlapping in cell $c$ from the same
    direction.}
    \label{fig:int-ind-cases1-2}
  \end{figure}
  Otherwise, there must be a vertex in $\tilde{e}
  \setminus \bar{e}$ that is contained in only one of $e$ and
  $e'$, say, in $e$, and we find $F_1$ as a subhypergraph; see
  Fig.~\ref{fig:int-ind-case1-g}. Now, assume that $\bar{e}$ is overlapping with
  $\tilde{e}$ coming from the right. If $\bar{e} \cap
  \tilde{e}$ contains a vertex of only one of the sets, say, $e$, we
  consider $\bar{e} \setminus \tilde{e}$. If there is a vertex not in
  $e$ (and not in $e'$), we have found $M_2$; see
  Fig.~\ref{fig:int-ind-case1-h}. (If $\tilde{e} \cap \bar{e}$
  contains a vertex of $e \cap e'$, we find $M_1$ instead). Otherwise, since
  $\tilde{e} \cup \bar{e}$ must overlap with at least one more edge,
  we can continue to explore more edges. As long as there is a
  hyperedge overlapping with the hyperedges starting from
  $\tilde{e}$ and $\bar{e}$ to the right, we choose the one ending
  rightmost, thus forming a path of hyperedges that is extending to
  the right. If this process eventually finds a vertex that is neither
  in $e$ nor in $e'$, we find $M_k$ as a subhypergraph; see
  Fig.~\ref{fig:int-ind-case1-i}.
  \begin{figure}[tb]
    \begin{subfigure}[t]{.32\textwidth}
      \centering
      \includegraphics[scale=1.1,page=7]{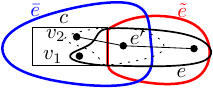}
      \caption{$F_1$.}
      \label{fig:int-ind-case1-g}
    \end{subfigure}
    \begin{subfigure}[t]{.32\textwidth}
      \centering
      \includegraphics[scale=1.1,page=8]{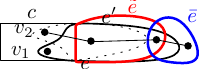}
      \caption{$M_2$.}
      \label{fig:int-ind-case1-h}
    \end{subfigure}
    \hfill
    \begin{subfigure}[t]{.32\textwidth}
      \centering
      \includegraphics[scale=1.1,page=9]{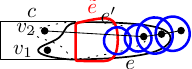}
      \caption{$M_k$.}
      \label{fig:int-ind-case1-i}
    \end{subfigure}
    \caption{$e$ and $e'$ overlapping in cell $c$ from the same
    direction.}
    \label{fig:int-ind-cases1-3}
  \end{figure}
  If we do not reach a vertex not in $e$ or $e'$ with the path because
  there are no more edges overlapping from the right, we know that
  there must be an edge overlapping the whole path from the left
  (otherwise, the edges of the path would not have been inserted
  before). Let $\bar{e}'$ be this edge. Now, if there is a vertex of
  $e \cap e'$ not contained in $\bar{e}'$, we have found $C_3$ as a
  subhypergraph; see Fig.~\ref{fig:int-ind-case1-j}. Otherwise, the
  part of the path outside of $\bar{e}'$ contains a vertex that is
  only in one of the hyperedges, say, in $e$. Then, the
  forbidden subhypergraph that we find is $F_k$ with $\bar{e}'$ and
  $e$ as the big hyperedges; see Fig.~\ref{fig:int-ind-case1-k}.
  \begin{figure}[tb]
    \begin{subfigure}[t]{.32\textwidth}
      \centering
      \includegraphics[scale=1.2,page=10]{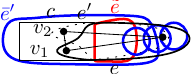}
      \caption{$C_3$.}
      \label{fig:int-ind-case1-j}
    \end{subfigure}
    \hfill
    \begin{subfigure}[t]{.32\textwidth}
      \centering
      \includegraphics[scale=1.2,page=11]{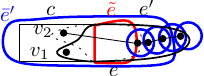}
      \caption{$F_k$.}
      \label{fig:int-ind-case1-k}
    \end{subfigure}
    \caption{$e$ and $e'$ overlapping in cell $c$ from the same
    direction.}
    \label{fig:int-ind-cases1-4}
  \end{figure}

  Now, we can consider the second case in which we get a contradiction
  to the interval property after inserting $e'$: Again, let $e'$ split
  a cell $c$ into $c_1$ and $c_2$ as before. Then, $e$ contains
  vertices from the cell left of $c$, at least one vertex $v_2$ of
  $c_2$, but there is also a vertex $v_1 \in c_1 \setminus e$, i.e.,
  $e$ does not completely contain $c_1$. We know that there must be at
  least one edge containing cell $c$. First, assume that such an
  edge~$\tilde{e}$ exists and there are vertices $v \in e \setminus
  \tilde{e}$ and $v' \in e' \setminus \tilde{e}$. Then, we find
  $M_1$ as a subhypergraph; see Fig.~\ref{fig:int-ind-case2-a}.

  If no such $\tilde{e}$ exists, we know that any edge containing
  $c$ must fully contain at least one of $e$ and $e'$ as a subset. On
  the other hand, we know that there must be at least one edge
  overlapping with $e$ and one edge overlapping with $e'$ (and by now,
  these two edges must be different). Assume that there are edges
  $\tilde{e}$ overlapping with $e$ and fully containing $e'$ and
  $\bar{e}$ overlapping with $e'$ and fully containing $e$.
	\begin{figure}[t]
    \begin{subfigure}[t]{.35\textwidth}
      \centering
      \includegraphics[page=2]{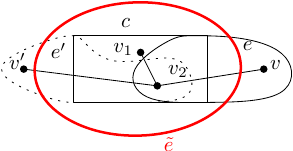}
      \caption{$M_1$.}
      \label{fig:int-ind-case2-a}
    \end{subfigure}
    \hfill
    \begin{subfigure}[t]{.45\textwidth}
      \centering
      \includegraphics[page=3]{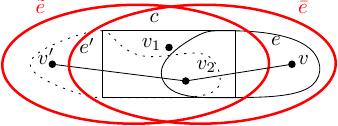}
      \caption{$F_1$.}
      \label{fig:int-ind-case2-b}
    \end{subfigure}
    \caption{$e$ and $e'$ overlapping in cell $c$ from different
    directions.}
    \label{fig:int-ind-cases2-1}
  \end{figure}
	Then we
  find $F_1$ as a subhypergraph; see Fig.~\ref{fig:int-ind-case2-b}.
	
	Now, assume that there is only a hyperedge $\tilde{e}$ overlapping
  with $e$ and fully containing $e'$; among these edges let
  $\tilde{e}$ be the one ending leftmost and (among the ones ending
  leftmost) the shortest one. We know that there must be at
  least one edge overlapping with $e'$, but no such edge can go to
  the right (and contain $c$), otherwise we would be in one of the
  previous cases. Let $\bar{e}$ be the hyperedge overlapping with
  $e'$ and ending leftmost. If $\bar{e}$ contains a vertex not
  contained in $\tilde{e}$, then we have found $M_2$ as a
  subhypergraph; see Fig.~\ref{fig:int-ind-case2-c}. Otherwise, we continue searching for
  the leftmost starting hyperedge overlapping with $\bar{e}$, forming
  a path of hyperedges reaching to the left. If eventually we reach at
  a vertex not contained in $\tilde{e}$, then we have found an
  $M_k$-subhypergraph; see Fig.~\ref{fig:int-ind-case2-d}. On the other hand, if the
  path ends before reaching out of $\tilde{e}$, by considering the
  union of the path hyperedges starting from $\bar{e}$, we know that
  there must be a hyperedge overlapping from the right. Due to the
  choice of $\tilde{e}$, this hyperedge may or may not overlap with
  $e$, but it must contain an element of $e$ that is not contained in
  $\tilde{e}$. Therefore, we find $F_k$ as a subhypergraph; see
  Fig.~\ref{fig:int-ind-case2-e}.
  \begin{figure}[tb]
    \begin{subfigure}[t]{.49\textwidth}
      \centering
      \includegraphics[page=4, scale=.95]{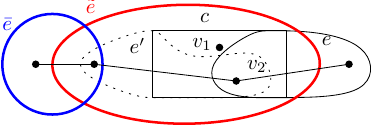}
      \caption{$M_2$.}
      \label{fig:int-ind-case2-c}
    \end{subfigure}
    \hfill
    \begin{subfigure}[t]{.49\textwidth}
      \centering
      \includegraphics[page=5, scale=.95]{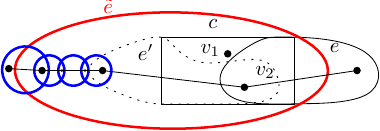}
      \caption{$M_k$.}
      \label{fig:int-ind-case2-d}
    \end{subfigure}
    \caption{$e$ and $e'$ overlapping in cell $c$ from different
    directions.}
    \label{fig:int-ind-cases2-2}
  \end{figure}

  In the remaining case, each edge containing cell~$c$ must fully
  contain both $e$ and $e'$. Let $\tilde{e}$ be the edge containing
  $c$ that is shortest and starts leftmost. Both for $e$ and
  $e'$ we know that there is at least one edge previously inserted
  that overlaps with them. Similarly to the argument before, we can
  start with the leftmost overlapping for $e'$ and the rightmost for
  $e$ and build paths of overlapping edges into these directions until
  we reach a vertex outside of $\tilde{e}$, or we find no further
  hyperedge to extend the path. If both paths leave $\tilde{e}$, we
  find an $M_{k}$-subhypergraph; see Fig.~\ref{fig:int-ind-case2-f}. Now, assume only
  \begin{figure}[tb]
    \begin{subfigure}[t]{.49\textwidth}
      \centering
      \includegraphics[page=6, scale=.9]{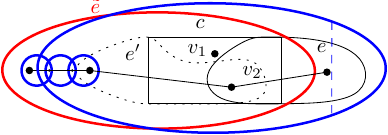}
      \caption{$F_k$.}
      \label{fig:int-ind-case2-e}
    \end{subfigure}
    \hfill
    \begin{subfigure}[t]{.49\textwidth}
      \centering
      \includegraphics[page=7, scale=.9]{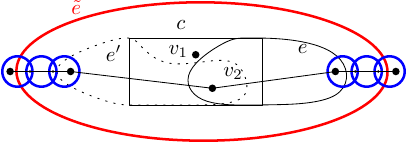}
      \caption{$M_k$.}
      \label{fig:int-ind-case2-f}
    \end{subfigure}
    \caption{$e$ and $e'$ overlapping in cell $c$ from different
    directions.}
    \label{fig:int-ind-cases2-3}
  \end{figure}
  the one for $e$ reaches out of $\tilde{e}$, but the one for
	$e'$ doesn't (the other case is symmetric). Since the hyperedges
	of the path for $e'$ have been
  inserted, there must still be a hyperedge overlapping with them. The only
  remaining possibility is then that this hyperedge $\bar{e}$ extends
  to the right and contains $c$ and fully contains both $e$ and $e'$.
  Due to the choice of $\tilde{e}$ being the shortest hyperedge
  containing $c$, $\bar{e}$ must also contain at least one cell right
  of $\tilde{e}$. Hence, we find an $F_k$-subhypergraph; see
  Fig.~\ref{fig:int-ind-case2-g}. The remaining
  case is that neither path reaches out of $\tilde{e}$. Then, apart
  from $\bar{e}$, with the symmetric argument we find a hyperedge
  $\bar{e}'$ that overlaps with the path for $e$, fully contains
  $e$ and $e'$, and reaches out of $\tilde{e}$ to the left. By using
  $\bar{e}'$ in place of $\tilde{e}$, we again find an
  $F_k$-subhypergraph; see Fig.~\ref{fig:int-ind-case2-h}. This completes the proof.
  \begin{figure}[tb]
    \begin{subfigure}[t]{.49\textwidth}
      \centering
      \includegraphics[page=8, scale=.88]{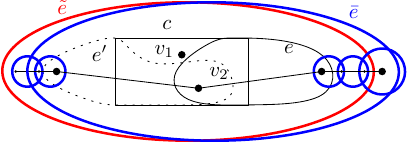}
      \caption{$F_k$.}
      \label{fig:int-ind-case2-g}
    \end{subfigure}
    \hfill
    \begin{subfigure}[t]{.49\textwidth}
      \centering
      \includegraphics[page=9, scale=.88]{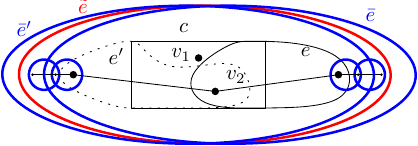}
      \caption{$F_k$.}
      \label{fig:int-ind-case2-h}
    \end{subfigure}
    \caption{$e$ and $e'$ overlapping in cell $c$ from different
    directions.}
    \label{fig:int-ind-cases2-4}
  \end{figure}
\end{proof}

The lemma shows that we can keep refining the cell-structure by
inserting edges that contain vertices of at least two different cells.
We end up with a cyclic order of cells so that each edge of the
connected component that we did not insert lies completely within a
single cell. Several edges can lie within the same cell, sharing
vertices, and forming a small hypergraph that imposes further
restrictions on the relative order of vertices within the cell.
However, the cell contains fewer than $d$ vertices. Hence, this small
hypergraph cannot contain any long cycles and, since we removed all
other forbidden subhypergraphs, must be an interval hypergraph.

\begin{lemma}
  If for any two adjacent cells there is a hyperedge containing the
  vertices of both cells, we can find a cycle as a subhypergraph.
  \label{lemma:cell-neighboring-cycle}
\end{lemma}
\begin{proof}
  We start at an arbitrary cell $c$. There must be a hyperedge
  $e_1$ containing both $c$ and the next cell in clockwise order. We
  iteratively form a path by considering the rightmost cell explored
  so far and finding a hyperedge that contains that cell as well as
  the cell right of it. Since the number of cells is finite, we
  eventually reach the first cell. By dropping edges fully contained
  in other edges found, if necessary, we have a complete cycle.
\end{proof}

\begin{lemma}
  In any interval hypergraph that is obtained from the connected
  component there is at least one pair of neighboring cells so that
  all edges containing both cells have been removed
  \label{lemma:cell-neighboring-break}
\end{lemma}
\begin{proof}
  The lemma is a direct corollary from
  Lemma~\ref{lemma:cell-neighboring-cycle} since if there is no such
  pair of cells, the condition of that lemma holds.
\end{proof}

\begin{lemma}
  Given a cyclic cell-order, let $c$ and $c'$ be a neighboring pair
  of cells in clockwise order. Removing all edges that contain both
  $c$ and $c'$ results in an interval hypergraph.
  \label{lemma:cell-neighboring-cycle-destruction}
\end{lemma}
\begin{proof}
  We number the cells $c' = c_1, c_2, \ldots, c_{k} = c$ in clockwise
  order. Next, we place the vertices on a straight line so that
  vertices of each cell form an interval on the line and the cells
  appear as $c_1, \ldots, c_k$ from top to bottom. Since the edges
  falling completely within a cell form an interval hypergraph, we put
  the vertices within a cell into an order that supports this
  interval hypergraph; recall that this is an interval hypergraph of
  constant size. Hence, each edge falling within a cell is supported.

  Now consider an edge $e$ that spans over several cells. If the
  interval that $e$ spans is over cells $c_i, \ldots, c_j$ with
  $1 \le i < j \le n$, it is supported by our order of vertices. On the
  other hand, if the cyclic interval of $e$ is of the type
  $c_i \ldots, c_k, c_1, \ldots, c_j$ with $1 \le j < j \le k$, then
  $e$ also contains the cells $c = c_k$ and $c' = c_1$ and, therefore,
  has been removed.
\end{proof}

\section{Interval Hypergraph Edge Deletion -- Implementation in
$O(m^2)$ Time}
\label{sec:hyperedge-deletion-runtime}
The first phase of our algorithm consists mainly
of searching for given subhypergraphs. In general, searching for a
subhypergraph of parameterized size $k$ is hard to achieve in time
$n^{o(k)}$ since this includes the hard search for
$k$-cliques~\cite{Chen20061346}.
However, the structure of our problem allows us  to do the search in
$O(m^2)$ time as follows. First, we check for cycles by considering
any edge $e$, choosing any pair $v_1, v_2$ of its up to $d$ vertices
(we have to try every pair), removing all edges containing both
vertices, and then trying to find a shortest path from $v_1$ to
$v_2$ using breadth-first search. If there is such a path of length
$k \le d$, we have found $C_{k+1}$, and we remove all its edges. Since
any edge has to be considered only once---it is then either removed or
cannot be part of a short cycle---this part takes $O(m^2)$ time.

For destroying the remaining types of forbidden subhypergraphs, we
make use of the fact, that each of them contains an edge that contains
all but 1 ($O_2$ and $F_k$), 2 ($M_k$), or 3 ($O_1$) vertices of the
subhypergraph. We try each edge $e$ to play that role. Since
$e$ has at most $d$ vertices (constant), we can try each
combination of its vertices for the vertices of the forbidden
subhypergraph in the edge as shown in
Fig.~\ref{fig:forbidden-subhypergraphs}. Since there are only up to
three more vertices not in $e$ required, we could try all combinations
for these and end up with an $O(m^2 n^3)$-time algorithm. However, we
can get rid of the factor $n^3$ as follows. Suppose there
are vertices $v_1, v_2 \in e$ and hyperedges $e_1, e_2$ so that
$v_1 \in e_1, v_2 \in e_2$ but $v_1 \notin e_2$ and $v_2 \notin e_1$.
If there is a vertex $v \in \left( e_1 \cap e_2 \right) \setminus e$ in
the intersection of $e_1$ and $e_2$ outside of $e$, then $v, v_1, v_2$
with the hyperedges $e, e_1$, and $e_2$ form a $C_3$-subhypergraph;
however, we have already removed short cycles, a contradiction.

Now, consider the search for $O_1$. If for each of the three involved
vertices in the larger hyperedge $e$ we find a hyperedge containing
vertices not in $e$, then we must have found $O_1$, otherwise the
above argument yields $C_3$. For the other forbidden subhypergraphs we
must additionally check whether there is at least one hyperedge
realizing exactly each of the necessary pairwise adjacencies within
$e$. For $M_k$, $k \le d-1$, this suffices to check for an occurrence.
For $O_2$ we must also check whether there is a hyperedge containing
the two nonadjacent vertices of $e$ and an element not in $e$. For
$F_k$, $k \le d-2$, we need a vertex in the intersection of a
hyperedge $e_1$ connecting the rightmost path-vertex to something
outside of $e$ with the second hyperedge $e_2$ containing $k+2$
vertices. This can be checked in $O(m)$ time by searching all feasible
hyperedges and marking vertices outside of $e$ if they lie in one such
vertex. Note that no hyperedge realizing one of the pairwise
adjacencies of $F_k$ can contain such a vertex of $e_1 \cap e_2
\setminus e$ since our above argument yields $C_3$ in that case.

Summing up, we can test in $O(m)$ time whether a given edge is the
``large edge''---the edge of highest cardinality---of any of the
forbidden subhypergraphs in $O(m)$ time. Since
after considering an edge it is either removed, or we know that it is
not contained as large edge in any forbidden subhypergraph, we can
make $\hyper$ \ffree in $O(m^2)$ time.

Then, we determine the connected components in linear time, find a
cycle for each of them and initialize the cell order, in $O(n+m)$ time
in total. For all components, the
stepwise refinement can be done in $O(m^2)$ time in total. Counting
the numbers of hyperedges between adjacent cells, determining the
optimum splitting point, as well as finding the final order, can all
be done in linear time (since the size of edges is constant).

\begin{repeattheorem}{thm:int-hyper-edge-removal-approx-runtime}
  We can find a $(d+1)$-approximation for \textsc{Interval Hypergraph Edge
  Deletion} on hypergraphs with $m$ hyperedges of rank~$d$ in $O(m^2)$ time.
\end{repeattheorem}

\section{Open Problems}
\label{app:openProblems}
While our paper yields insight into the complexity of several aspects
of SBCM, several interesting problems remain open.
\begin{itemize}
	\item Does the greedy algorithm yield an approximation for 2-SBCM?
		Can it be reasonably generalized to more than
		two characters per meeting? Can we find an optimal starting
		permutation in polynomial time?
	\item It is open if there always is an optimal solution for 2-SBCM
		that uses at most one block crossing between two meetings when the
		start permutation is \emph{not} fixed.
		Our experiments strongly suggest some relations between $n$, $k$
		and the optimum in random instances, but we have not properly
		investigated this.
	\item Can we get better results for any variant of the problem if
		we consider the start permutation part of the input and fixed?
	\item Can similar approximation results be obtained for simple
		crossings rather than block crossings? Since our analysis and
		algorithms heavily depend on the extended powers of block
		crossings, it seems hard to adjust our approach.
\end{itemize}

\end{document}